\documentclass[11ptUSenglish,onecolumn]{article}
\usepackage[a4paper, total={6in, 8in}]{geometry}

\usepackage[utf8]{inputenc}
\usepackage{natbib}
\usepackage{times}
 
\usepackage{hyperref}
\usepackage{xcolor}		
\definecolor{darkblue}{rgb}{0, 0, 0.5}
\hypersetup{colorlinks=true,citecolor=darkblue, linkcolor=darkblue, urlcolor=darkblue}
\usepackage{amsmath}
\usepackage{graphicx}
\usepackage{tikz}
\usepackage{amssymb}
\usetikzlibrary{arrows}
\usepackage{hyperref} 
\usepackage{amsthm}
\usepackage{float}
\usepackage{mathrsfs} 
\usepackage[final]{pdfpages}
\usepackage[normalem]{ulem}
\usepackage{setspace}
\usepackage{enumitem}   
\usepackage{algorithm,algorithmic}
\usepackage{caption}

\makeatletter

      \theoremstyle{plain}
      \newtheorem{assumption}{Assumption}
      \newtheorem{theorem}{Theorem}
       \newtheorem{lemma}{Lemma}
       \newtheorem{condition}{Condition}
       \newtheorem{remark}{Remark}

\newcommand{\expit}{\text{expit}}

\newcommand\@erelb@r[1]{%
  \mathrel{\tikz[baseline=-.5ex]\draw[#1] (0,0)--(0.3,0);}
}

\newcommand{\erelbar}[1]{\@erelbar#1}
\def\@erelbar#1#2{%
  \ifcase\numexpr#1*4+#2\relax
    \@erelb@r{-}\or     
    \@erelb@r{->}\or    
    \@erelb@r{-|}\or    
    \@erelb@r{->|}\or   
    \@erelb@r{<-}\or    
    \@erelb@r{<->}\or   
    \@erelb@r{<-|}\or   
    \@erelb@r{<->}\or   
    \@erelb@r{|-}\or    
    \@erelb@r{|->}\or   
    \@erelb@r{|-|}\or   
    \@erelb@r{|<->|}\or 
    \@erelb@r{|<-}\or   
    \@erelb@r{|<->}\or  
    \@erelb@r{|<-|}\or  
    \@erelb@r{|<->|}    
  \else
    \@wrong
  \fi
}
\makeatother

\newcommand{\gop}{
\overset{p}{\to}
}

\newcommand{\gol}{
\overset{\mathcal{L}}{\to}
}

\newcommand{\fd}[3]{\frac{\partial }{\partial #1} #2\bigg|_{#3}}
\newcommand{\sds}[3]{\frac{\partial^2 }{\partial #1^2} #2\bigg|_{#3}}

\newcommand{\indep}{\perp \!\!\! \perp}

\newlength\myindent
\newcommand\bindent{%
  \begingroup
  \setlength{\itemindent}{\myindent}
  \addtolength{\algorithmicindent}{\myindent}
}
\newcommand\eindent{\endgroup}

\newcommand{\expan}[2]{
E_n
\fd{\beta}{\left(\frac{v}{E_n r} - \gamma kl\right)}{\substack{\beta=#1\\ b=#2}}
&\approx
E_n\frac{\partial}{\partial\beta}\left(\frac{v}{E_n r} - \gamma kl\right) \bigg|_{\substack{\beta=\beta_{0,\gamma} \\ b=b_0}}\\
&+  (#1-\beta_{0,\gamma})^TE_n \frac{\partial^2}{\partial \beta^2}\left(\frac{v}{E_n r} - \gamma kl\right)\bigg|_{\substack{\beta=\beta_{0,\gamma}\\b=b_0}}\\
&+
(#2-b_0)^T E_n\fd{b}
{\frac{\partial}{\partial \beta}\left(\frac{v}{E_n r} - \gamma kl\right)}
{{\substack{\beta=\beta_{0,\gamma}\\b=b_0}}}
}

\newcommand{\AC}{\mathcal{A}}

\newcommand{\fdr}[2]{\frac{\partial }{\partial #1} #2}
\newcommand{\sdsr}[2]{\frac{\partial^2 }{\partial #1^2} #2}
\newcommand{\sddr}[3]{\frac{\partial }{\partial #1}\frac{\partial}{\partial#2} #3}

\newcommand{{\citerelativesparsity}}{\cite{relsparSIM}}
\newcommand{{\citeinference}}{our work}
\usepackage{authblk}

  \title{\huge Inference for relative sparsity
  }
\author[1,2]{Samuel J. Weisenthal}
\author[1]{Sally W. Thurston}
\author[1]{Ashkan Ertefaie}
\affil[1]{Department of Biostatistics and Computational Biology, University of Rochester School of Medicine and Dentistry, Rochester, NY}
\affil[2]{Medical Scientist Training Program, University of Rochester School of Medicine and Dentistry, Rochester, NY}
\affil[ ]{\textit {\{Samuel\_Weisenthal,\ 
Sally\_Thurston,\
Ashkan\_Ertefaie\}@URMC.Rochester.edu}}

\begin{document}











\maketitle

 \begin{abstract}
{In healthcare, there is much interest in estimating policies, or mappings from covariates to treatment decisions. Recently, there is also interest in constraining these estimated policies to the standard of care, which generated the observed data. A relative sparsity penalty was  proposed to derive policies that have sparse, explainable differences from the standard of care, facilitating justification of the new policy.  However, the developers of this penalty only considered estimation, not inference. 
Here, we develop inference for the relative sparsity objective function, because characterizing uncertainty is crucial to applications in medicine.   Further, in the relative sparsity work, the authors only considered the single-stage decision case; here, we consider the more general, multi-stage case.
Inference is difficult, because the relative sparsity objective depends on the unpenalized value function, which is unstable and has infinite estimands in the binary action case. Further, one must deal with a non-differentiable penalty.
To tackle these issues, we nest a weighted Trust Region Policy Optimization function within a relative sparsity objective, implement an adaptive relative sparsity penalty, and propose a sample-splitting framework for post-selection inference. We study the asymptotic behavior of our proposed approaches, perform extensive simulations, and analyze a real, electronic health record dataset.}
\end{abstract}


\section{Introduction}

Treatment policies, or mappings from patient covariates to treatment decisions, can help healthcare providers and patients make more informed, data-driven decisions, and there is great interest in both the  statistical and reinforcement learning communities  in developing methods for deriving these policies \citep{Chakraborty2013,futoma2020popcorn,uehara2022review}. 
  There is particularly recent interest in deriving constrained versions of these policies. While there has been work on the general theory of constrained  reinforcement learning \citep{le2019batch,geist2019theory}, several methodologies can be more specifically  categorized as ``behavior-constrained'' policy optimization, an umbrella term used in \cite{wu2019behavior} to encompass the  array of methods that constrain the new, suggested policy to be similar to the ``behavioral'' policy that generated the data.   Examples of `behavior-constrained' policy optimization  include  entropy-constrained policy search  \citep{haarnoja2017reinforcement,ziebart2008maximum,peters2010relative};  what \cite{le2019batch} calls ``conservative''   policy improvement methods, such as guided search and Trust Region Policy Optimization (TRPO)   \citep{levine2014learning,schulman2015trust,schulman2017proximal,achiam2017constrained,le2019batch}, which constrain optimization such that large divergences  from the previous policy are discouraged; and other approaches with similar goals, such as likelihood weighting, entropy penalties  \citep{fujimoto2019off,ueno2012weighted,dayan1997using, peters2010relative,haarnoja2017reinforcement,ziebart2008maximum},  imitation learning   \citep{le2016smooth}, value constrained model-based reinforcement learning  \cite{futoma2020popcorn,farahmand2017value}, and tilting  \citep{kennedy2019nonparametric,kallus2020efficient}. 
In {\citerelativesparsity}, a relative sparsity penalty was developed, which differs from behavior constraints in existing studies 
in that it focuses on explainability and relative interpretability between the suggested policy and the standard of care. 
In {\citerelativesparsity}, however, only estimation, not inference, was considered for the relative sparsity objective function.  In {\citeinference}, therefore, we consider the challenging problem of inference for the relative sparsity objective. Further, in    {\citerelativesparsity}, the authors only considered the single-stage decision setting; here, we consider the more general, multi-stage decision setting.
The objective function in relative sparsity  combines a raw value (expected reward) function with a relative Lasso penalty, where both components pose challenges for inference.   In the binary action case, under a parameterized policy, the raw value function is optimized by estimands that are infinite or arbitrarily large in magnitude, a consequence of the the fact that the policy that solves the raw value objective function is deterministic \citep{lei2017actor,puterman2014markov,relsparSIM}. 

The contribution of {\citeinference} beyond the existing literature, can be summarized in the following five ways. First, to address the issue of estimands that are of infinite or arbitrarily large magnitude, we propose a double behavior constraint, nesting a weighted TRPO behavior constraint within the relative sparsity objective.
Second, based on work in \cite{zou2006adaptive}, we develop  methodology for an adaptive relative sparsity formulation,  which improves discernment of the penalty.   Third, we provide a sample splitting framework that is free of issues  associated with post-selection inference \citep{cox1975note,leamer1974false,kuchibhotla2022post}. Fourth, we rigorously study the asymptotic properties of all frameworks: inference for existing TRPO methods has not been studied due to the general focus on pure prediction in robotics applications. To fill this gap, we develop
novel theory for inference in the TRPO framework and, in particular, for the weighted \citep{thomas2015safe,mcbook} TRPO estimator. Further, we rigorously study the asymptotic  theory for the adaptive relative sparsity penalty and develop theory for confidence intervals in the sample splitting setting.  We take special care to develop theory around the nuisance, which appears not only in the denominator of the inverse probability weighting expression  but, also, in the sample splitting case, within the suggested policy itself.
Fifth, 
we consider the more general multi-stage, Markov Decision Process (MDP), setting.
We  perform simulation studies, revealing how the magnitude of the tuning parameter  impacts inference, and showing where the proposed methodology and theory succeeds, and where it might fail.  
We conclude {\citeinference} with a data analysis of a real, observational dataset, derived from the MIMIC III database \citep{johnson2016mimic2, johnson2016mimic1, goldberger2000physiobank}, performing inference on a relatively sparse decision policy for vasopressor administration.
Although similar, routinely collected health data has been used for prediction (e.g., \citep{futoma2015comparison,lipton2015learning,weisenthal2018predicting}), there has been less work toward developing rigorous decision models with this data, and we fill this gap as well.

Developing the statistical inference properties of  the relative sparsity penalty allows us to better port this useful technique to healthcare, where the uncertainty associated with the new, suggested policy is important in order to guide healthcare providers, patients, and other interested parties as they choose whether or not to adopt these new treatment strategies.
This work ultimately facilitates translation of  data-driven treatment strategies from the  laboratory to the clinic,  where they might substantially improve health outcomes.   

\section{Notation}
 \label{sec:not}
 Throughout {\citeinference},  we  use subscript $0$ and $n$ to denote a true parameter and an estimator derived from a sample of size $n$ (i.e., $\beta_n$ is an estimator for $\beta_0$ based on a sample of size $n$), respectively.
If we have a vector, $v_t,$ indexed at time $t,$ we index  dimension $k$ as $v_{t,k}.$  We index the components of a parameter vector, $\theta,$ as $\theta_k$. In many cases, we will further subscript parameters according to penalty tuning parameters (e.g., $\lambda$), as in $\theta_{\gamma,\lambda}$, and, in this case, we will index dimension $k$ as $\theta_{\gamma,\lambda,k}$.
Let $\beta$ denote a parameter that indexes an arbitrary policy, $\pi_{\beta}(a|s),$ where $a$ and $s$ are a binary action (treatment) and state (patient covariates), respectively. Let similarly $b$ denote an arbitrary  nuisance parameter that indexes the existing, behavioral policy, which corresponds to the standard of care.  Let  $E,E_n$ be the true and empirical expectation operators, where, for some function $f,$ $E_n f = \frac{1}{n} \sum_i f(X_i).$ 
 We will sometimes write, for arbitrary  functions  $f$ and $g$, \[E_n \frac{f}{E_n g}=\frac{1}{n}\sum_i \frac{f_i}{E_n g}=\frac{1}{n}\sum_i \frac{f_i}{\frac{1}{n}\sum_i g_i}.\] 
We will often use capital letters to refer to an average and their lowercase counterparts to refer to the elements being averaged; i.e., we write $F_n=E_n f.$
Let $\pi_{\beta,b}$ denote a policy in which some components of $\beta$ are fixed to components in $b$.  Assuming random variable $A$ is discrete and random variable $S$ is continuous, let $E_b(f)=\sum_a \int_s f(a,s)\pi_b(a|s)p(s)ds$ denote an expectation of some deterministic (non-random) function $f$ with respect to policy $\pi_b.$
 Note that, in line with our subscript conventions for estimands and estimators mentioned above, $\pi_{\beta_n}$ is an estimator for $\pi_{\beta_0}$.

\section{Background}
\label{rs.background.multi.notation}

\subsection{Markov decision processes (MDPs)}
\label{inf:backgroundmdp}
Consider the multi-stage, discrete-time, Markov decision process (MDP), a general model for data that evolves over time based on the actions of some agent as it interacts with an environment \citep{bellman1957markovian}.  The MDP and its extensions have been used to model many problems in the medical domain; for examples related to diabetes, hypotension, and mobile health, see \cite{Chakraborty2013,ertefaie2018constructing,futoma2020popcorn,lei2017actor,luckett2019estimating}. We aim to address similar healthcare problems here. Let us have a continuous, $K$-dimensional state, $S\in \mathbb{R}^K$, which may, e.g., contain a patient's covariates.
 Let us also have a binary action, $A\in\{0,1\}$ which may be, e.g., the administration of a medication.
Let us sample $n$ independent and identically distributed length-($T+1$) patient trajectories of states and actions (hence, all trajectories must be of the same length).  
The random sample from a single  trajectory is then of the form
$\left\{S_{i,0},A_{i,0},S_{i,1},A_{i,1},\dots,S_{i,T},A_{i,T},S_{i,T+1}\right\}_{i=1,\dots,n}$,
where $S_{i,t}$ is the stage-$t$ state of patient $i$ and $A_{i,t}$ is the stage-$t$ action for patient $i.$   A trajectory is sampled from a fixed distribution denoted by $P_0$, which can be factored into an initial state distribution, $P_0(S_0)$, the  transition probability, $P_{0,t+1}(S_{t+1}|A_t,S_t,\dots,A_0,S_0),$ and the data-generating policy, 
$P_{0,t}(A_t|S_t,\dots,A_0,S_0).$  The latter is called a ``policy,'' because we can imagine that a trajectory is constructed by a healthcare provider (and/or patient) drawing an action conditional on the patient history.
We assume that the policy is Markov and does not change over time, which are common assumptions in these problems \citep{sutton2018reinforcement}. 
\begin{assumption}{Markov property:}
\label{assum:markov}
 $P_{0,t}(A_t|S_t,\dots,A_0,S_0)=P_{0,t}(A_t|S_t).$
\end{assumption}
\begin{assumption}{Stationarity:}
\label{assum:stationarity}
 $P_{0,t}(A_t|S_t)=P_{0}(A_t|S_t).$
\end{assumption}
To facilitate interpretation, we parameterize $P(A_t|S_t)$ with arbitrary vector parameter $\theta$ (which can refer to $\theta=\beta$ or $\theta=b$ depending on whether we are referring to the parameter that we are optimizing over or the behavioral policy parameter),  so that it is  $P_{\theta}(A_t|S_t).$ We then denote $P_{\theta}(A_t|S_t)$ as $\pi_{\theta}(A_t|S_t),$ as is convention \citep{sutton2018reinforcement}. For this parameterization, we propose the model
\begin{align}
\label{eq:thetapolicyInf}
&\pi_{\theta}(A_t=1|S_t=s)=\expit(\theta^T s)=
\frac{\exp({\theta^Ts})}{1+\exp({\theta^Ts})}.
\end{align}
Under Assumption \ref{assum:stationarity}, Assumption \ref{assum:markov}, and the model in (\ref{eq:thetapolicyInf}), we have the sampling distribution \[P_0(S_0)\pi_{b_0}(A_0|S_0)\prod_{t=1}^{T}\pi_{b_0}(A_t|S_t)P_{0,t}(S_t|A_{t-1},S_{t-1},\dots,A_0,S_0),\] where $P_0(A_t|S_t)$,  is parameterized by $b_0$ under (\ref{eq:thetapolicyInf}) and becomes $\pi_{b_0}(A_t|S_t)$.

Let there also be a deterministic, stationary reward function 
\begin{equation}
\label{eq:R}
R(S_t,A_t,S_{t+1}),
\end{equation}
which maps the state at some time to a utility.
 The total return for one trajectory (or episode) is the sum of  rewards from the first time step to the end of the trajectory
$\sum_{t=0}^{T}R(S_t,A_t,S_{t+1}).$ Because we focus on  finite-horizon cases, there is no need to consider discounted cumulative rewards, in which the contribution of states that occur later in time is down-weighted \citep{sutton2018reinforcement}.  

Let us also parameterize an arbitrary policy $\pi_{\beta}$ with a vector of coefficients $\beta$  using (\ref{eq:thetapolicyInf}).
Define the value of a policy, which is the expected return when acting under the policy $\pi_{\beta}$, as
\begin{equation}
\label{eq:V0}
V_0(\beta)=E_{\beta}\sum_{t=0}^{T}R(S_t,A_t,S_{t+1}).
\end{equation}
The reward-maximizing policy  $\pi_{\beta_0}$ can be obtained by solving
\begin{equation}
\label{eq:beta0inference}
\beta_0=\arg\max_{\beta}V_{0}(\beta).
\end{equation} 
\begin{remark}
\label{rem:beta0determ}
One cannot perform inference for $\beta_0.$ In Lemma \ref{lemma:det} (Appendix \ref{supp:proverDeterm}), we show that, under the model for the policy in  (\ref{eq:thetapolicyInf}), because the optimal policy $\pi_{\beta_0}$ is deterministic,  $\beta_0$ diverges in magnitude.     
\end{remark}
The divergence issue precludes inference in 
{\citerelativesparsity}, where the authors optimize the ``relative sparsity'' objective,
$V_0(\beta)-\lambda||\beta-b_{0}||_1.$
We overcome this issue by replacing 
 the relative sparsity  ``base'' objective, $V_0,$ with the full TRPO objective, as we will now describe.

\subsection[On inference with TRPO]{On inference with Trust Region Policy Optimization (TRPO)}
\label{sec:behcon}
 
Inference for Trust Region Policy Optimization (TRPO) \citep{schulman2015trust} has not previously been considered, to our knowledge, because the robotics applications for which TRPO was developed might not benefit from inference in the way that medical applications might. The remainder of this  section, therefore, contains what we believe to be novel insights.

The objective $V_0$ in (\ref{eq:V0}) is not behavior constrained, which precludes inference. To mitigate this issue, we add a Kullback-Leibler ($KL$) \citep{kullback1951information} behavior constraint \citep{schulman2015trust}.   We choose $KL$ divergence because it is an expectation and therefore has favorable asymptotic properties. 
 More specifically, for fixed $\gamma,$ we will employ the following objective as a new base objective for relative sparsity,
\begin{equation}
\label{eq:M0}
    M_{0}(\beta,b,\gamma) = V_0(\beta) - \gamma KL_0(\beta,b),
\end{equation}
where 
\begin{equation}
\label{eq:KL0}
KL_0(\beta,b)=E_b\log\left({\prod_{t=0}^T\pi_{b}(A_t|S_t)}\biggr/{\prod_{t=0}^T\pi_{\beta}(A_t|S_t)}\right).
\end{equation}
In general, finding a policy with minimal $KL$ divergence from the behavioral policy is equivalent to finding a policy that maximizes likelihood \citep{van2000asymptotic}. In the objective function ({\ref{eq:M0}}), this would be achieved by setting $\gamma=\infty.$

Now, the following estimand, $\beta_{0,\gamma},$ is \textit{behavior-constrained}
\begin{equation}
\label{eq:beta0gamma}
\beta_{0,\gamma}=\arg\max_{\beta}M_0(\beta,b_0,\gamma).
\end{equation}
One can perform inference for $\beta_{0,\gamma},$ which will be finite in magnitude, even though, as discussed in Remark \ref{rem:beta0determ}, one cannot perform inference for $\beta_0,$ because $\beta_0$ is infinite in magnitude.
\begin{remark}
\label{rem:finitebeta0gamma}
If $\gamma>0,$ the behavior-constrained solution $\beta_{0,\gamma}=\arg\max_{\beta} M_0(\beta,b_0,\gamma)$ is finite in magnitude, which allows for inference.  
\end{remark}
 Our estimand, $\beta_{0,\gamma}$, depends on a parameter $\gamma;$ we are targeting a behavior-constrained estimand. This is in contrast to typical penalization, where one targets some estimand that does not depend on a parameter and penalizes in order to reduce the variance of the estimator. 
Having addressed the issue of infinite estimands, we can now augment the behavior constrained objective function in (\ref{eq:M0}) with a relative sparsity objective.
 
\section{Methodological Contributions}
 \label{sec:contributions}
\subsection[Adding relative sparsity to TRPO]{Adding (adaptive) relative sparsity to  Trust Region Policy Optimization (TRPO)}
\label{sec:method:adaptiveRe}

Having stated the Trust Region Policy Optimization (TRPO) objective in (\ref{eq:M0}), we will now add to it relative sparsity.  As discussed in Remark \ref{rem:finitebeta0gamma},  maximizing the TRPO objective function defined by (\ref{eq:M0}) allows for inference for a behavior-constrained policy. Behavior constraints create closeness to the standard of care, which facilitates adoption, since changes to practice guidelines can pose challenges for healthcare providers \citep{gupta2017physician,lipton2018mythos,rudin2019stop}.  However, in a healthcare setting, one must convince the healthcare provider (and possibly also the patient) to adopt a new treatment policy. This is facilitated if the number of parameters that differ between the two policies is small, which is related to considerations such as cognitive burden \citep{miller2019explanation,du2019techniques,relsparSIM}.
To achieve relative sparsity, we propose
\begin{align}
\label{eq:W0}
\notag   W_0(\beta,\beta_{0,\gamma},b_0)&= M_0(\beta,b_0)-\lambda\sum_{k=1}^K{w}_{0,k} |\beta_k-b_{0,k}|\\
   &= V_0(\beta)-\gamma KL_0(\beta,b_0)-\lambda\sum_{k=1}^K{w}_{0,k} |\beta_k-b_{0,k}|, 
\end{align}
where $M_0$ is defined in (\ref{eq:M0}) and
$w_{0,k}=1/|\beta_{0,\gamma,k}-b_{0,k}|^{\delta}.$
Accordingly, we define our estimand as
\begin{equation}
\label{eq:beta0gammalambda}
\beta_{0,\gamma,\lambda}=\arg\max_{\beta} W_0(\beta,\beta_{0,\gamma},b_0).
\end{equation} 
\newcommand{\deltameaning}{
In \cite{zou2006adaptive}, the authors try $\delta=\{.5, 1, 2\}.$
Recall that the weight of the penalty term is $w_n=1/|\beta_{n,\gamma}-b_n|^{\delta}.$
If $\delta=1,$ then the weight for coefficient $k$ is $|(\beta_{n,\gamma,k})-(b_{n,k})|,$ and coefficients that are truly non-behavioral, $|(\beta_{0,\gamma,k})-(b_{0,k})|>0$, will have a smaller penalty while the coefficients that are truly behavioral, $|(\beta_{0,\gamma,k})-(b_{0,k})|=0$, will have a larger penalty  (\cite{zou2006adaptive} notes that this is related to Breiman's non-negative Garotte). If $\delta>1,$ if  $|(\beta_{n,\gamma,k})-(b_{n,k})|>1,$ then the penalty weight will the inverse of the power of a number greater than one, and it will therefore be small.  If $|(\beta_{n,\gamma,k})-(b_{n,k})|<1,$ the penalty weight will be the inverse of a power of a number less than one, and it will therefore be large. Hence, setting $\delta>1$ should make the penalty more ``adaptive,'' and, as $\delta\rightarrow\infty,$ the adaptivity should increase. If $\delta<1,$ and $|(\beta_{0,\gamma,k})-(b_{0,k})|>1,$ then the denominator of the penalty will be less large, and the penalty will be larger than it would be if $\delta=1$; in the reverse case,  if $|(\beta_{0,\gamma,k})-(b_{0,k})|<1,$ the denominator will be larger, and the penalty will be smaller  than it would be if $\delta=1.$  Hence, the penalty will be tempered for all coefficients (the adaptive penalty becomes less adaptive).  As $\delta \rightarrow 0$ we approach  $|(\beta_{n,\gamma,k})-(b_{n,k})|^0=1,$ in which case the penalty weight does not depend on the coefficients at all (it is completely non-adaptive), and we return to a standard relative Lasso.
}
The added Lasso penalty brings relative sparsity to our behavior constrained estimand $\beta_{0,\gamma}.$  In practice, increasing the weight of the Lasso penalty will also cause some behavior-constraint, but this is intended to  be minimal; unlike in the objective function of {\citerelativesparsity}, where the Lasso penalty jointly performs shrinkage to behavior and selection, the Lasso penalty in (\ref{eq:W0}) should only perform selection, while the $KL_0$ penalty performs shrinkage to behavior. In Equation (\ref{eq:W0}), The degree of shrinkage and selection will be controlled, respectively, by the tuning parameters $\gamma,$ which controls the degree of closeness to the behavioral policy, and $\lambda,$ which controls the degree of relative sparsity. Further, $\delta,$ which is proposed in \cite{zou2006adaptive}, controls the adaptivity of the adaptive Lasso penalty (and is discussed more in Appendix \ref{app:deltameaning}).  We will discuss how to choose these tuning parameters when we discuss estimation.


\subsection{Sample splitting in the relative sparsity framework}
\label{sam.split}

Let $\AC$ be a set containing the indices of selected (non-behavioral) covariates.  Let $1_{\AC}$ be an indicator for the selected (non-behavioral) covariates, so
 $1_{\mathcal{A}}=(1_{1\in \mathcal{A}},\dots,1_{K \in \mathcal{A}})^T$.  The Lasso penalty in (\ref{eq:W0}) performs selection, giving us $\AC.$  However, as with any selection, we must avoid issues with post-selection inference \citep{leamer1974false}. For this, we use sample splitting \citep{cox1975note}, where we perform selection on one split and then inference on a second, independent split.

We first optimize (\ref{eq:W0}) to obtain a selection, $\AC$. In standard sample splitting, one would eliminate the non-selected coefficients. In our case, we keep non-selected variables, but we fix their parameters to their behavioral counterparts, and then we perform inference only with respect to the non-behavioral parameters.  For this purpose, letting $\odot$ denote element-wise multiplication, we propose a novel representation of a policy as 
\begin{equation}
\label{eq:postSelectPolicy}
\pi_{\beta,b}(1|s)=\expit(\beta^T (s\odot 1_{\mathcal{A}})+b^T (s\odot (1_K-1_{\mathcal{A}})) ),
\end{equation}
where $1_K$ is a length-$K$ vector of ones.
We can then take partial derivatives of $\pi_{\beta,b}$ with respect to $\beta$ or $b,$ as necessary, while fixing some entries of $\beta$ to $b_0$ (in practice, we fix these entries not to $b_0$ but to an estimator of $b_0$).  The form of each partial derivative, amended for the  post-selection policy in (\ref{eq:postSelectPolicy}), is included in Appendix \ref{app:gradients}.  In particular, the cross derivatives  now have extra terms, because the nuisance, $b,$ appears in the suggested policy as well as in the behavioral policy.


\section{Estimation}

\subsection{Value}
\label{sec:estValue}
We now discuss estimation of the value, or expected return, under some arbitrary policy, indexed by $\beta$.
For this, we will need to take a counterfactual expectation, which can be done using importance sampling or inverse probability weighting \citep{kloek1978bayesian,precup2000eligibility,thomas2015safe,horvitz1952generalization,robins1994estimation,Chakraborty2013,precup2000eligibility}. 
 \newcommand{\IS}{
 Define \[\left\{A_{i,0:T},S_{i,0:T+1}\right\}=\
\{A_{i,0},\dots,A_{i,T},S_{i,0},\dots,S_{i,T+1}\}.\]
Define the likelihood of a trajectory under $P$ and the behavioral policy $\pi$ indexed by an arbitrary parameter $b$ is 
\[
P_{b}(A_{0:T},S_{0:T+1})=P_0(S_0)\pi_{b}(A_0|S_0)P_{b}(A_{1:T},S_{1:T})P_{0,T+1}(S_{T+1}|A_{T},S_{T},\dots,A_0,S_0),
\]
where 
\[P_{b}(A_{1:T},S_{1:T})=\prod_{t=1}^{T}\pi_{b}(A_t|S_t)P_{0,t}(S_t|A_{t-1},S_{t-1},\dots,A_0,S_0).\]
We obtain the potential value under policy $\pi_{\beta}$ by noting that
\begin{align*}
V_{0}(\beta)&=E_{\beta}\left[\sum_{t=0}^{T}R(S_t,A_t,S_{t+1})\right]\\&=E_{b}\left\{\rho(A_{0:T},S_{0:T+1};\beta,b) \sum_{t=0}^{T}R(S_t,A_t,S_{t+1})\right\},
\end{align*} where
\begin{align*}
\rho(A_{0:T},S_{0:T+1};\beta,b)&=\frac{P_{\beta}(A_{0:T},S_{0:T+1})}{P_{b}(A_{0:T},S_{0:T+1})}\\
&=\frac{P(S_0)\pi_{\beta}(A_0|S_0)P_{\beta}(A_{1:T},S_{1:T})P_{T+1}(S_{T+1}|A_{T},S_{T},\dots,A_0,S_0)}
{P(S_0)\pi_{b}(A_0|S_0)P_{b}(A_{1:T},S_{1:T})P_{T+1}(S_{T+1}|A_{T},S_{T},\dots,A_0,S_0)}\\
&=\frac{\pi_{\beta}(A_0|S_0)P_{\beta}(A_{1:T},S_{1:T})}
{\pi_b(A_0|S_0)P_{b}(A_{1:T},S_{1:T})}\\
&=\frac{\pi_{\beta}(A_0|S_0)\prod_{t=1}^{T}\pi_{\beta}(A_t|S_t)P_t(S_t|A_{t-1},S_{t-1},\dots,A_0,S_0)}
{\pi_b(A_0|S_0)\prod_{t=1}^{T}\pi_b(A_t|S_t)P_t(S_t|A_{t-1},S_{t-1},\dots,A_0,S_0)}\\
&=\frac{\pi_{\beta}(A_0|S_0)\prod_{t=1}^{T}\pi_{\beta}(A_t|S_t)}
{\pi_b(A_0|S_0)\prod_{t=1}^{T}\pi_b(A_t|S_t)}\\
&=\prod_{t=0}^{T}\frac{\pi_{\beta}(A_t|S_t)}{\pi_b(A_t|S_t)}.
\end{align*}
} 
  Starting  with expressions that only depend on the observed data, we rederive, in Appendix \ref{est.v0.Precup}, the well-known fact \citep{thomas2015safe} that an estimand for the potential value under a policy $\pi_{\beta}$, can be written, as long as the denominator is never zero (which is formalized as an assumption in Section \ref{sec:assumptions}), as 
\begin{align}
\label{eq:V0times1}
V_0(\beta)=E_{\beta}\sum_{t=0}^{T}R(S_t,A_t,S_{t+1})
=E_{b}\left\{\frac{\prod_{t=0}^T \pi_{\beta}(A_t|S_t)}{\prod_{t=0}^T \pi_{b}(A_t|S_t)}\ \sum_{t=0}^{T}R(S_t,A_t,S_{t+1})\right\}.
\end{align}
Then $V_0$ can be estimated using an inverse probability weighted estimator.
However, in the multi-stage case, the vanilla inverse probability weighted estimator is unstable.  We therefore use a weighted, as it is called in \cite{thomas2015safe}, or self-normalized, as it is called in  \cite{mcbook}, importance sampling estimator,
\begin{align}
\label{eq:vnestw}
{V}_n(\beta,b)
= \frac{\frac{1}{n}\sum_{i=1}^n \frac{\prod_{t=0}^T \pi_{\beta}(A_{i,t}|S_{i,t})}{\prod_{t=0}^T \pi_{b}(A_{i,t}|S_{i,t})}\sum_{t=0}^T R(S_{i,t},A_{i,t},S_{i,t+1})}{\frac{1}{n}\sum_{i=1}^n\frac{\prod_{t=0}^T \pi_{\beta}(A_{i,t}|S_{i,t})}{\prod_{t=0}^T \pi_{b}(A_{i,t}|S_{i,t})}.}
\end{align}

Note that (\ref{eq:vnestw}) has two arguments while (\ref{eq:V0}) has one, because (\ref{eq:vnestw}) depends on $b$ through the empirical inverse probability weighting ratio, whereas the non-empirical terms involving $b$ cancel in (\ref{eq:V0}). Let \[b_n=\arg\max_{b} \sum_{i=1}^n\sum_{t=0}^T\log \pi_{b}(A_{i,t}|S_{i,t})\] denote an estimator of $b_0$ (see Appendix \ref{app:estimatingBehavioral}). 
Then, replace $b$ in (\ref{eq:vnestw}) with $b_n$  to obtain $V_n(\beta,b_n),$  an estimator for the potential value, $V_0(\beta).$
For the parameter of the (unrestricted) optimal policy, $\beta_0$, define the estimator  
$\beta_n=\arg\max_{\beta} V_n(\beta,{b_n}).$

\subsection{Trust Region Policy Optimization (TRPO)}

Now that we have shown how to estimate $V_0$ in (\ref{eq:V0}), we discuss how to estimate $M_0$ in (\ref{eq:M0}).  For this, we write an estimator of of $KL_0$, defined in (\ref{eq:KL0}),  as 
\begin{equation}
\label{eq:KLn}
KL_n(\beta,b)=\frac{1}{n}\sum_i \log \frac{\prod_{t=0}^T \pi_{b}(A_{i,t}|S_{i,t})}{\prod_{t=0}^T \pi_{\beta}(A_{i,t}|S_{i,t})}.
\end{equation}
We showed how we can use ${V}_n$ to estimate $V_0$ in (\ref{eq:vnestw}). 
We hence estimate $M_0,$ defined in (\ref{eq:M0}), with 
\begin{align}
\label{eq:mn}
{M}_n(\beta,b_n,\gamma) &= {V}_n(\beta,b_n)-\gamma KL_n(\beta,b_n).
\end{align} 
 We estimate $\beta_{0,\gamma},$ defined in (\ref{eq:beta0gamma}), with 
 \begin{equation}
 \label{eq:betangamma}
 \beta_{n,\gamma}=\arg\max_{\beta}{M}_n(\beta,b_n,\gamma).
 \end{equation}
Increasing $\gamma$ increases the degree of behavior constraint of $\pi_{\beta_{0,\gamma}},$ which is important for obtaining ``closeness'' to the standard of care, as discussed in Section \ref{sec:behcon}. As mentioned in Section \ref{sec:contributions}, we benefit from closeness to the standard of care when we translate a suggested policy to the clinic, because  the suggested policy will be more likely to be adopted when its suggested treatment aligns with the established guidelines. 
Increasing $\gamma$ also  leads to stabilization of the objective, since $V_n$ is typically more unstable than $KL_n$, where minimizing $KL_n$, as discussed in Section \ref{sec:behcon}, is equivalent to maximizing likelihood.  This stabilizes inference as well as estimation.

\subsection{Adaptive relative sparsity}
\label{sec:adaptiveRelativeLasso}
We finally define an estimator for $W_0$, defined in (\ref{eq:W0}), as
\begin{align}
\label{eq:Wn}
 W_n(\beta,\beta_{n,\gamma},b_n)&={M}_n(\beta,b_n)-\lambda\sum_{k=1}^K w_{n,k} |\beta_k-b_{n,k}|,
\end{align}
where we estimate $w_{0,k},$ defined in Section \ref{sec:method:adaptiveRe}, with
$w_{n,k}={1}/{|\beta_{n,\gamma,k}-b_{n,k}|^{\delta}}.$
We then estimate $\beta_{0,\gamma,\lambda},$ defined in (\ref{eq:beta0gammalambda}), 
using
\begin{equation}
\label{eq:betangammalambda}
\beta_{n,\gamma,\lambda}=\arg\max_{\beta}W_n(\beta,\beta_{n,\gamma},b_n).
\end{equation}

\subsection{Tuning parameters}
\label{sec:estTuning}

We now discuss the three tuning parameters in Equation (\ref{eq:W0}): $\gamma,\lambda,$ and $\delta$. 
The tuning parameter $\gamma$ impacts the weight on the $KL$ divergence portion of the penalty in (\ref{eq:M0}) and will determine the closeness to the standard of care.   From an estimation standpoint, $\gamma$ impacts the stability of the objective function, and should be chosen based on
the stability of the estimation in a training dataset, which we will illustrate in simulations and in the real data analysis here. After selection,  in the post-selection inference step, there will be fewer free parameters, so the estimation stability caused by $\gamma$ will be even larger; hence, it is reasonable to choose a slightly smaller $\gamma$ but to still expect stability in inference.

Given $\gamma,$ one can choose $\lambda$ as 
\begin{equation}
\label{lambda.crit.diff.naught}
\lambda_0 = \max \{\lambda: V_0(\beta_{0,\lambda})\geq V^{min}\},
\end{equation}
where $V^{min} = V_0(b_0) + \sigma_V(b_0),$ and, if we use  ${V}_n,$  given in (\ref{eq:vnestw}), as an estimator for $V_0,$ then $\sigma_V$ is the asymptotic standard deviation of $\sqrt{n}{V}_n.$ 
We take a $\max$ in (\ref{lambda.crit.diff.naught}) to ensure maximum sparsity and closeness to behavior within the set of policies that have acceptable value of at least $V^{min}$. 
We  estimate $\lambda_0$ from (\ref{lambda.crit.diff.naught}) using 
\begin{equation}
\label{lambda.crit.diff}
\lambda_n = \max \{\lambda: {V}_n(\beta_{n,\lambda})\geq {V}_n^{min}\},
\end{equation}
where ${V}_n^{min} = {V}_n(b_n,b_n) + \sigma_{n,V}(b_n),$ and $\sigma_{n,V}$ now refers to standard error of ${V}_n$, where estimation of $\sigma_{n,V}$
is described in Appendix \ref{app:valuevar}. 
Note that  the standard error is not a general-purpose selection threshold, since the standard error will decrease with increasing sample size.  One might consider also the standard deviation of the behavioral value or a certain percentage increase in value.

The tuning parameter $\delta>0,$ which is proposed in \cite{zou2006adaptive}, impacts the adaptivity of the adaptive Lasso penalty, and increasing $\delta$ should, in theory, lead to a stronger penalty for the coefficients that are truly equal to their behavioral counterparts and a  weaker penalty for the coefficients that  truly diverge from their behavioral counterparts, as discussed in more detail in Appendix \ref{app:deltameaning}. 
The finite sample behavior of the adaptive Lasso penalty, however, is sometimes unpredictable, which has been discussed in e.g., \cite{potscher2009distribution}, so we recommend trying a few different values of $\delta$ in a training set, as is done in \cite{zou2006adaptive}, where the authors use $\delta\in \{.5, 1, 2\}.$

\section{Theory}
\label{theory}

We provide theory for inference for $\beta_{0,\gamma},$ which was defined in (\ref{eq:beta0gamma}). This is novel in its own right, since it applies to Trust Region Policy Optimization (TRPO) \citep{schulman2015trust} whose inferential properties have not been well characterized, and it further concerns a novel, weighted version of TRPO, 
defined in (\ref{eq:vnestw}). However, our overall goal is not to show results for TRPO, but to  operationalize this theory toward performing inference for relative sparsity in the post-selection setting. We also provide theory that will be used in the selection diagrams to visualize the variability our estimators.   

\subsection{Assumptions}
\label{sec:assumptions}
 We make the following causal identifiability assumptions. 

\begin{assumption}
 \leavevmode
\label{assum:pos} 
\begin{enumerate}[label=(\roman*)]

\item  {Positivity:} $\pi_{b_0}(A=a|S=s)>0\ \ \forall \ a,s.$

\item  {Consistency:}
$S_{t+1}(A_0,\dots,A_t)=S_{t+1}.$


\item {No interference:}
Let $S_{t+1}(a_0,\dots,a_t)$ to be the potential state under, possibly contrary to observation, $A_0,\dots,A_t$ being fixed to $a_0,\dots,a_t.$ If $i$ and $j$ index different patients, we have that
    \[S_{i,t+1}(a_{i,0},\dots,a_{i,t},a_{j,0},\dots,a_{j,t})=S_{i,t+1}(a_{i,0},\dots,a_{i,t}).\] 
  
\item {Sequential randomization:}
\[S_{t+1}(a_0,\dots,a_t),S_{t+2}(a_0,\dots,{a}_{t+1}),\dots,S_{T+1}(a_0,\dots,{a}_T)\indep A_t|{S}_t,{A}_{t-1}={a}_{t-1}.\]
\end{enumerate}

\end{assumption}
 
Under model \ref{eq:thetapolicyInf}, Assumption \ref{assum:pos} (i) implies that $|b|<\infty$ and, consequently, by Remark \ref{rem:finitebeta0gamma}, that $|\beta|<\infty$.
In our case, it is likely that Assumption \ref{assum:pos} (ii), consistency, and Assumption \ref{assum:pos} (iii), no intereference, are satisfied.  Assumption \ref{assum:pos} (iv), sequential randomization, can be more problematic, although we have included the same covariates that were used in the literature \cite{futoma2020popcorn}; one could conceivably adjust for more covariates in future work.  
 
 The following is necessary to establish asymptotic consistency.
\begin{assumption}
 \label{as:mContInBetabandcompact}   
Define $\bar{\mathcal{S}}=\mathcal{S}_0,\dots,\mathcal{S}_{T+1},$
$\bar{\mathcal{A}}=\mathcal{A}_0,\dots,\mathcal{A}_{T},$
$\bar{{S}}={S}_0,\dots,{S}_{T+1},$ and
$\bar{{A}}={A}_0,\dots,{A}_{T}.$
For the objective defined in (\ref{eq:m}), $m: \bar{\mathcal{S}}\times \bar{\mathcal{A}}\times B\times B'\mapsto \mathbb{R},$ where $\beta\in B$ and $b\in B',$ we have that $B\times B'$ is compact. Moreover, for all $(\beta,b)\in B\times B',$ we have that $m(\cdot,\beta,b)$ is Borel measureable on $\bar{\mathcal{S}}\times \bar{\mathcal{A}}$ and, for each $(\bar{S},\bar{A})\in \bar{\mathcal{S}}\times \bar{\mathcal{A}}$ we have that $m(\bar{S}, \bar{A},\cdot)$ is continuous on $B\times B'.$
\end{assumption}

\begin{assumption}
    \label{assum:Sbounded}
       The states are uniformly bounded; i.e., there exists $0<C<\infty$ such that $P(|S_t|\geq C)=0,$ for all $t.$ 
\end{assumption}

Quantities such as mean arterial blood pressure (MAP), creatinine, or urine output are physiologic quantities, and, therefore, random variables representing these quantities will be restricted to take on finite values.

\begin{assumption}
\label{assum:boundedR}
     The reward is bounded; i.e., $\left|R(s,a,s')\right|<\infty$  $\forall\ s, a,s'.$ 
\end{assumption}
Often, the reward will be based on state variables, which are themselves bounded by Assumption \ref{assum:Sbounded}.  It is therefore reasonable to assume boundedness of the reward (although this would be violated if one were to, e.g., include within the reward an infinite penalty for mortality). 

 Besides the Markov Decision Process (MDP) and causal assumptions (Assumption \ref{assum:stationarity} and Assumption \ref{assum:pos}), we also make assumptions that allow us to extend Theorem 5.41 in \cite{van2000asymptotic}, which assumes boundedness of  the partial derivatives of $M_0$. 
We now argue that these boundedness assumptions are reasonable, largely because of the causal assumption of positivity (Assumption \ref{assum:pos} (i)) and the
physiologic bounds on the state variables.

\begin{assumption}
\label{as:mbounded}
 Let $\zeta=(\beta,b)^T\in \mathbb{R}^{2K}.$ The following partial derivatives exist and, for any length-$T$ trajectory of states and actions, satisfy   \begin{align}
 \label{disp}
\frac{\partial^3}{\partial \zeta_i \zeta_j\zeta_k}\left(\frac{\prod_{t=0}^T \pi_{\beta}(A_t|S_t)}{\prod_{t=0}^T \pi_{b}(A_t|S_t)}\sum_{t=0}^{T}R(S_t,A_t,S_{t+1})-\gamma \log\left(\frac{\prod_{t=0}^T \pi_{b}(A_t|S_t)}{\prod_{t=0}^T \pi_{\beta}(A_t|S_t)}\right)\right) \leq \overset{\cdots}{\text{m}},
 \end{align}  for some integrable, measurable function $\overset{\cdots}{\text{m}},$ and for every $\zeta=(\beta,b)^T$ in a neighborhood of $\zeta_0=(\beta_{0,\gamma},b_0)^T.$ 
 \end{assumption}
  
 We will argue in the following section that $\overset{\cdots}{\text{m}}$ in Assumption \ref{as:mbounded} is a constant for the reinforcement learning problems that we are interested in solving, which depend on the reward and the states, both of which are usually bounded.  
 
 Let us consider the components of (\ref{disp}).
 We will start by discussing $\sum_{t=0}^{T}R(S_t,A_t,S_{t+1}).$
Note that since $T$ is finite, boundedness of the reward in Assumption \ref{assum:boundedR} implies that \[\left|\sum_0^T  R(s_t,a_t,s_{t+1})\right|<\infty, \forall (s_0,\dots,s_{T+1}, a_0,\dots,a_T).\] 
Moreover, because we are considering the binary action case, the policies are bounded above by 1 and below by 0. 
Following Assumption \ref{assum:pos} (i), we  have that the inverse of the product of the policies  is bounded (i.e., $\prod_{t=0}^T \pi_{b}(A_t|S_t)^{-1}<\infty$).  

The derivatives in (\ref{disp}) are various combinations of the reward, states, and policies and are therefore similarly bounded; one can see the forms of the first and second order derivatives in Appendix \ref{app:gradients}.  
Sometimes, a derivative of the logarithm of a policy appears.  In our generalized linear model setting, because of model (\ref{eq:thetapolicyInf}), we see that the differentiated logarithm turns into a generalized linear model score of the form 
$(a-\pi(A=1|s))s,$
which is bounded because the policies and state are bounded.  
The partial derivatives are over $T$ time steps, and, because $T$ is finite,  the summands or factors are bounded.

The $\log{\prod_{t=0}^T \pi_{b}(A_t|S_t)}\biggr/{\prod_{t=0}^T \pi_{\beta}(A_t|S_t)}$ term in (\ref{disp}) corresponds to the the Kullback-Leibler (KL) divergence between the suggested policy and the data-generating, behavioral policy, defined in (\ref{eq:KL0}), both of which are $\expit$ models according to (\ref{eq:thetapolicyInf}). We have that  minimizing KL divergence is equivalent to maximizing the log likelihood, as discussed in Section 5.5 of \cite{van2000asymptotic}. Hence, since the partial derivatives of the log likelihood are well behaved, we have that the partial derivatives of the $KL$ function are similarly well behaved. 

Uniqueness of the maximizer of an objective function, $M_0$, is important for establishing consistency of the maximizer of ${M}_n$.
\begin{assumption}
\label{as:uniqueMaxofM}
We have uniqueness of the maximizer of $M_0(\beta,b_0),$ defined in (\ref{eq:M0}); i.e., 
\[M_0(\beta_{0,\gamma},b_0)>M_0(\beta,b_0)\] for all $\beta\neq \beta_{0,\gamma}.$
\end{assumption}
Assumption \ref{as:uniqueMaxofM} could be relaxed to a local uniqueness assumption \citep{loh2013regularized,eltzner2020testing}.  
In some problems, if taking the action sequence $(A_1=1,A_2=0)$ and $(A_1=0,A_2=1)$ gives equivalent value, the policy that maximizes value may not be unique, but, in many problems of interest, the order of treatments matters.  
\begin{assumption}
\label{as:bconsist}
We assume that $b_n$ is consistent for $b_0,$ the data generating behavioral policy; i.e.  $b_n\gop b_0$. Moreover, we assume that the behavioral estimator  is $\sqrt{n}$-consistent; i.e.,
$\sqrt{n}(b_n-b_0)=O_{P}(1).$
\end{assumption}
Assumption \ref{as:bconsist} holds for the estimators we  use for the behavioral policy, assuming correct specification of the model in (\ref{eq:thetapolicyInf}). 



 Define the Jacobian (gradient), Hessian, and cross derivative of ${M}_n$ and $M_0,$ respectively, as
\begin{align}
\label{eq:derivsDefn}
    \notag &J_n=\fdr{\beta}{{M}_n} \in \mathbb{R}^K, H_n=\sdsr{\beta}{{M}_n} \in \mathbb{R}^{K\times K}, X_n=\sddr{b}{\beta}{{M}_n} \in \mathbb{R}^{K\times K},\\ 
    &J_0=\fdr{\beta}M_0\in \mathbb{R}^K, H_0= \sdsr{\beta}{M_0}\in \mathbb{R}^{K\times K}, X_0=\sddr{b}{\beta}{M_0}\in \mathbb{R}^{K\times K}.
\end{align}

\begin{assumption}
\label{as:hessianInvert}
We have that $H_0,$ which is defined in (\ref{eq:derivsDefn}), exists and is non-singular. 
\end{assumption}

Assumption \ref{as:hessianInvert} is standard and necessary to isolate  the policy coefficients from other terms in a Taylor expansion.  


\subsection{Preliminary remarks and results}

We first include some preliminary remarks and results to aid us in proving the theorems that will follow. 
\begin{remark}
\label{causal:ident}
Since the reward is a deterministic function of the states, under Assumptions \ref{assum:pos} (i)-(iv), one can make arguments similar to those presented in \cite{pearl2009causality,munoz2012population,van2018targeted,ertefaie2018constructing} and \cite{parbhoo2022generalizing}  to identify, as a function of the observed data, the potential value $V_0(\beta),$ which is the value we would obtain if we were to assign treatments based on the policy $\pi_{\beta}$. This corresponds to the unweighted version of the estimator in \ref{eq:vnestw}.
\end{remark}

\newcommand{\proveEr}{
For intuition, we will show first in the single stage case \[E r = E_b \frac{\pi_{\beta}}{\pi_b}=\int_s\sum_a \frac{\pi_{\beta}(a|s)}{\pi_b(a|s)}\pi_{b}(a|s)p(s)ds =\int_s\sum_a \pi_{\beta}(a|s)p(s)ds=1.\]
Now for the multi-stage case,
\begin{align*}
Er &= E_b \frac{\prod_{t=0}^T \pi_{\beta}(A_t|S_t)}{\prod_{t=0}^T \pi_{b}(A_t|S_t)}\\&=\int_{s_0} \sum_{a_0} \dots \int_{s_T} \sum_{a_T} \frac{\pi_{\beta}(a_0|s_0)\dots \pi_{\beta}(a_T|s_t)}{\pi_{b}(a_0|s_0)\dots \pi_{b}(a_T|s_t)} \pi_{b}(a_0|s_0)\\ &\ \ \ \times p(s_0)p(s_1|a_0,s_0)\dots P_t(s_t|a_{t-1},s_{t-1},\dots,a_0,s_0)\pi_{b}(a_T|s_t)ds_0\dots ds_T \\
&=\int_{s_0} \sum_{a_0} \dots \int_{s_T} \sum_{a_T}  \pi_{\beta}(a_0|s_0)p(s_0)p(s_1|a_0,s_0)\dots \pi_{\beta}(a_T|s_t)\\ & \ \ \ \times P_t(s_t|a_{t-1},s_{t-1},\dots,a_0,s_0)ds_0\dots ds_T \\
&=1.
\end{align*}
}

\begin{remark}
\label{rem:Er1}
As shown in \cite{thomas2015safe}, 
we have that \[E_n {\prod_{t=0}^T \pi_{\beta}(A_t|S_t)}\biggr/{\prod_{t=0}^T \pi_{b}(A_t|S_t)} \gop E {\prod_{t=0}^T \pi_{\beta}(A_t|S_t)}\biggr/{\prod_{t=0}^T \pi_{b}(A_t|S_t)}\] by the Law of Large Numbers. Note further that $E{\prod_{t=0}^T \pi_{\beta}(A_t|S_t)}\biggr/{\prod_{t=0}^T \pi_{b}(A_t|S_t)}=1,
$ as  shown in Appendix \ref{app:proveEr}. 
\end{remark}

\begin{remark}
\label{rem:Vnconsist}
By Slutsky's theorem, since $E_n {\prod_{t=0}^T \pi_{\beta}(A_t|S_t)}\biggr/{\prod_{t=0}^T \pi_{b}(A_t|S_t)} \gop 1$ by Remark \ref{rem:Er1}, we have that ${V}_n\gop V_0$. This is also shown in \cite{thomas2015safe}.
\end{remark}

The following consistency statements need to be verified, because we are using weighted importance sampling, as shown in  (\ref{eq:vnestw}).
\begin{lemma}
\label{lemma:MnConsist}
We have that ${M}_n$ is consistent for $M_0,$ or that ${M}_n\gop M_0.$
\end{lemma}
\begin{proof}
See Appendix \ref{app:prf:lemma:MnConsist}
\end{proof}

\begin{lemma}
\label{lem:limDerivsRn}
We have that \[E_n\fdr{\beta}{{\prod_{t=0}^T \pi_{\beta}(A_t|S_t)}\biggr/{\prod_{t=0}^T \pi_{b}(A_t|S_t)}}\gop 0,\] \[E_n \fdr{b}{{\prod_{t=0}^T \pi_{\beta}(A_t|S_t)}\biggr/{\prod_{t=0}^T \pi_{b}(A_t|S_t)}}\gop 0,\] \[E_n \sddr{b}{\beta}{{\prod_{t=0}^T \pi_{\beta}(A_t|S_t)}\biggr/{\prod_{t=0}^T \pi_{b}(A_t|S_t)}}\gop 0.\]
\end{lemma}
\begin{proof}
 See Appendix \ref{app:limDerivs}.
\end{proof}
We will use Lemma \ref{lem:limDerivsRn} to show that the partial derivatives of ${M}_n,$ defined in (\ref{eq:mn}), converge to the partial derivatives of $M_0$, defined in (\ref{eq:M0}), which requires verification because $M_n$ contains a weighting term, as shown in  (\ref{eq:vnestw}).  For the same reason, we verify the following Lemma.

\begin{lemma}
\label{lem:limDerivsMn}
We have that $J_n \gop J_0, H_n \gop H_0,$ and $X_n\gop X_0.$
\end{lemma}
\begin{proof}
    See Appendix \ref{app:lem:limDerivsMn:prf}
\end{proof}

\subsection[Weighted TRPO asymptotics]{Weighted Trust Region Policy Optimization (TRPO) consistency and asymptotic normality}
We will now ensure that  $\beta_{n,\gamma}$ is well behaved asymptotically. Recall that $\beta_{n,\gamma},$  is the maximizer of ${M}_n,$ which is defined in (\ref{eq:mn}). Then $M_n$ serves as the ``base'' objective of the double-penalized relative sparsity objective, $W_n,$ which is defined in (\ref{eq:Wn}). We need, essentially, that when the relative sparsity penalty goes away (in the adaptive Lasso case), or is taken away (in the sample splitting, post-selection case), we are left with an objective function that gives us a maximizer, $\beta_{0,\gamma}$, that is amenable to inference.

To show consistency of $\beta_{n,\gamma},$ after making extensions to take into account the weighting term in the importance sampling objective (\ref{eq:vnestw}), we can apply 
results from \cite{wooldridge2010econometric} on two-step M-estimators.  For asymptotic normality of $\beta_{n,\gamma},$   
we extend a classical proof, versions of which can be found in e.g. \cite{van2000asymptotic} and \cite{wooldridge2010econometric}.

\begin{theorem}  \leavevmode
\label{thm:betanConsistNorm}
\vspace{1pt}
    \begin{enumerate}[label=(\roman*)]    
    \item We have 
consistency of the TRPO estimator $\beta_{n,\gamma}$ for $\beta_{0,\gamma}$, or that
$\beta_{n,\gamma}\gop \beta_{0,\gamma}$.
\item We have asymptotic normality for the weighted TRPO estimator, $\beta_{n,\gamma}$, or that
\[
\sqrt{n}(\beta_{n,\gamma}-\beta_{0,\gamma})\gol -(H_0)^{-1}
\left(z_0+ X_0q_0\right),
\]
where, scaling the gradient of (\ref{eq:mn}), which is defined in (\ref{eq:derivsDefn}), and taking a limit, 
\small
\[\sqrt{n}E_n\frac{\partial}{\partial \beta} \left (\frac{\frac{\prod_{t=0}^T \pi_{\beta}(A_t|S_t)}{\prod_{t=0}^T \pi_{b}(A_t|S_t)}\sum_{t=0}^TR(S_t,A_t,S_{t+1})}{E_n \frac{\prod_{t=0}^T \pi_{\beta}(A_t|S_t)}{\prod_{t=0}^T \pi_{b}(A_t|S_t)}}-\gamma \log\left(\frac{\prod_{t=0}^T \pi_{b}(A_t|S_t)}{\prod_{t=0}^T \pi_{\beta}(A_t|S_t)}\right)\right )\bigg|_{\substack{\beta=\beta_{0,\gamma}\\b=b_0}}\gol z_0,\]
\normalsize
$\sqrt{n}(b_n-b_0) \gol q_0,$
and 
the Hessian, $H_0,$ and the cross derivative, $X_0,$ are defined in (\ref{eq:derivsDefn}).
Asymptotic normality of $\beta_{n,\gamma}$ then follows from the fact that $H_0$ and $X_0$ are constants and $z_0$ and $q_0$ are normally distributed random variables.
    \end{enumerate}
\end{theorem}
\begin{proof}
   For Theorem \ref{thm:betanConsistNorm} (i), see Appendix \ref{app:consistproof}. For Theorem \ref{thm:betanConsistNorm} (ii), see Appendix \ref{app:norprf}.
\end{proof}


\newcommand{\consistproof}{
To show consistency of $\beta_{n,\gamma},$ or that $\beta_{n,\gamma}\gop \beta_{0,\gamma}$, we can apply 
results from \cite{wooldridge2010econometric} on two-step M-estimators. 
By \cite{wooldridge2010econometric}, Section 12.4.1, to show consistency of  $\beta_{n,\gamma}=\arg\max_{\beta}{M}_n(\beta,b_n)$ for $\beta_{0,\gamma},$ we must first have that 
 $b_n\gop b_0,$ which is true by Assumption \ref{as:bconsist}. Given this, we then
must show  the following conditions:
\begin{condition}
\label{cond:UWLLN}
We have that ${M}_n(\beta,b)$ satisfies the Uniform Weak Law of Large Numbers over $B\times B',$ where $\beta\in B$ and $b\in B',$ or that $\sup_{\beta,b}|{M}_n(\beta,b)-M_0(\beta,b)|\gop 0,$
\end{condition}
\begin{condition}
\label{cond:Ident}
We have that $M_0(\beta_{0,\gamma},b_0)>M_0(\beta,b_0)$ for all $\beta \in B$ such that $\beta\neq \beta_{0,\gamma}$.
\end{condition}
We will start by showing Condition \ref{cond:UWLLN}.
Note that $V_n$ and ${M}_n$ are the weighted objectives in (\ref{eq:vnestw}) and (\ref{eq:mn}), and define $\tilde{V}_n=E_n v$ and $\tilde{M}_n=E_n m$ as their unweighted counterparts, where $v$ and $m$ are defined in (\ref{eq:v}) and (\ref{eq:m}). 
To show Condition \ref{cond:UWLLN}, note that
\begin{align*}
\sup_{\beta,b}|{M}_n(\beta,b)-M_0(\beta,b)|
&=\sup_{\beta,b}|{V}_n(\beta,b)-\gamma KL_0-M_0(\beta,b)|\\
&=\sup_{\beta,b}\left|{V}_n(\beta,b)-\tilde{V}_n(\beta,b)+\tilde{V}_n(\beta,b)-\gamma KL_0-M_0(\beta,b)\right|\\
&=\sup_{\beta,b}\left|\frac{E_n v}{E_n r}-\frac{E_n v}{E r}+\frac{E_n r}{E r}-\gamma KL_0-M_0(\beta,b)\right|\\
&=\sup_{\beta,b}\left|\frac{E_n r}{E r} - \frac{E_n v(E_nr -E r)}{E_n rEr }-\gamma KL_0-M_0(\beta,b)\right|\\
&\leq \sup_{\beta,b}\left|\frac{E_n r}{E r} -\gamma KL_0-M_0(\beta,b)\right| + \sup_{\beta,b}\left|\frac{E_n v(E_nr -E r)}{E_n rEr }\right|\\
&= \sup_{\beta,b}\left|\tilde{M}_n(\beta,b)-M_0(\beta,b)\right| + \sup_{\beta,b}\left|\frac{E_n v}{E_nr}(E_nr -E r)\right|\\
&\leq \sup_{\beta,b}\left|\tilde{M}_n(\beta,b)-M_0(\beta,b)\right| + \sup_{\beta,b}\left|\frac{E_n v}{E_nr}\right|\sup_{\beta,b}\left|E_nr -E r\right|\\
&=I+II.
\end{align*}
Where we have used that $E_r=1,$ which was shown in Remark \ref{rem:Er1}. Theorem 12.1 in \cite{wooldridge2010econometric} states that if $\beta\in B$ and $b\in B'$ for  compact subsets of $\mathbb{R}^K,$ $B$ and $B',$  and  if $m(\beta,b)$ is a measurable, real-valued function that is continuous and bounded on $B\times B',$ then $\sup_{\beta,b}|E_n m - E m|\gop 0.$ By Assumptions \ref{as:mContInBetabandcompact} and \ref{as:mbounded}, $m$ satisfies these conditions.  Hence, we have  \[\sup_{\beta,b}\left|\tilde{M}_n(\beta,b)-M_0(\beta,b)\right|=\sup_{\beta,b}\left|E_n m-Em\right|\gop 0.\]  Hence (I) vanishes.
 For (II), Set $R(S_t,A_t,S_{t+1})=1/T,$ where $R$ is defined in (\ref{eq:R}), and $\gamma=0,$  and note that \[m=v-\gamma \log kl = rR(S_t,A_t,S_{t+1})-\gamma \log kl=r,\] for $r$ defined in (\ref{eq:r}). Note then that Assumptions \ref{as:mContInBetabandcompact} and \ref{as:mbounded}, which apply to $m,$ apply equally to $r,$ which gives us that $\sup_{\beta,b}|E_nr - Er|\gop 0$ by Theorem 12.1 in \cite{wooldridge2010econometric}. Note that  $|b|<\infty$ by Assumption \ref{assum:pos} (i) on positivity, and hence $\sup_{\beta,b}|E_nv/E_nr|$ is bounded for all $\beta\in B$ and $b\in B'.$ Hence, we have that (II) vanishes.
Hence, because (I) and (II) vanish, we have uniform convergence of ${M}_n$. 
We further have that Condition \ref{cond:Ident}, the identification condition, holds by Assumption \ref{as:uniqueMaxofM}. 
We therefore conclude, based  on \cite{wooldridge2010econometric}, that
  the two-step M-estimator, $\beta_{n,\gamma},$ is consistent for $\beta_{0,\gamma}$.  
}

\newcommand{\lemErgrad}{
Note that by the Law of Large Numbers, $E_n \fdr{\beta}{r}\gop E \fdr{\beta}{r},$ so we just need to show that $E \fdr{\beta}{r}=(0,\dots,0)^T.$
Note that 
\begin{align*}
E \fdr{\beta}{r}
&=E_b\fdr{\beta}r\\
&=E_b \fdr{\beta}\frac{\prod_{t=0}^T \pi_{\beta}(A_t|S_t)}{\prod_{t=0}^T \pi_{b}(A_t|S_t)}\\
&=E_b\left(\frac{\prod_{t=0}^T \pi_{\beta}(A_t|S_t)}{\prod_{t=0}^T \pi_{b}(A_t|S_t)} \fdr{\beta}{\log \prod_{t=0}^T \pi_{\beta}(A_t|S_t)}\right)\\
&=E_{\beta}\left(\fdr{\beta}{\log \prod_{t=0}^T \pi_{\beta}(A_t|S_t)}\right)\\
&=E_{\beta}\left(\sum_{t=0}^T \fdr{\beta}{\log \pi_{\beta}(A_t|S_t)}\right)\\
&=E_{\beta}\left(\sum_{t=0}^T (A_t - \pi_{\beta}(A_t=1|S_t))S_t\right)\\
&=\sum_{t=0}^T E_{\beta}((A_t - \pi_{\beta}(A_t=1|S_t))S_t)\\
&=\sum_{t=0}^T E(E_{\beta}((A_t - \pi_{\beta}(A_t=1|S_t))S_t)|S_t))\\
&=\sum_{t=0}^T E((E_{\beta}(A_t|S_t) - \pi_{\beta}(A_t=1|S_t))S_t))\\
&=\sum_{t=0}^T E((\pi_{\beta}(A_t=1|S_t) - \pi_{\beta}(A_t=1|S_t))S_t))=0.
\end{align*}
where we have made use of the forms of the gradients that are defined in Appendix \ref{app:gradients}.
 }

\newcommand{\norprf}{
By Theorem \ref{thm:betanConsistNorm} (i) and Assumptions \ref{as:mbounded}-\ref{as:bconsist}, we can perform a multivariate Taylor expansion of the gradient of the objective around the point $(\beta_0,b_0).$ We will now be extending Theorem 5.41 of \cite{van2000asymptotic} to the case where we have a nuisance parameter and a weighting $E_nr.$  We will not write the higher order derivatives, as they will still disappear in the limit, as in the proof of Theorem 5.41. The terms in which the higher order derivatives appear contain either $\beta_{n,\gamma}-\beta_{0,\gamma}$ or $b_n-b_0,$ the first of which is $o_P(1)$ by Theorem \ref{thm:betanConsistNorm} (i) and the second of which is $o_P(1)$ by Assumption \ref{as:bconsist}, and the derivatives themselves are averages (as noted for the use case in \cite{van2000asymptotic} as well), and hence $O_P(1)$ by the Central Limit Theorem, so they terms in which they appear will be $o_P(1)O_P(1)=o_P(1)$.   

We have therefore an expansion just concerning the terms of interest,
\begin{align*}
\expan{\beta}{b}
\end{align*}
We can write this succinctly as 
\begin{equation}
\label{eq:earlyExp}
J_n(\beta,b) \approx J_n(\beta_{0,\gamma},b_0)+H_n{(\beta-\beta_{0,\gamma})} +  X_n{(b-b_0)},
\end{equation}
where $J_n=E_n z,$ for $z$ defined in (\ref{eq:z}), 
$H_n,$ and $X_n$ refer to the empirical gradient, Hessian, and cross derivative, respectively, defined in  (\ref{eq:derivsDefn}), and we only show the arguments for the gradient in (\ref{eq:earlyExp}) to distinguish its form on the left-hand side of the equation from its form on the right-hand side of the equation.

Since our objective function, $E_n \left({v}/{E_n r} - \gamma kl\right),$ given $b_n$, is maximized by $\beta_{n,\gamma},$ we can write (\ref{eq:earlyExp}) as
\[0=J_n(\beta_{n,\gamma},b_n)\approx J_n(\beta_{0,\gamma},b_0)+H_n(\beta_{n,\gamma}-\beta_{0,\gamma}) +  X_n(b_n-b_0).\]
Rearranging, 
\begin{equation*}
\label{eq:anotherExp}
-(J_n(\beta_{0,\gamma},b_0)+ X_n(\beta_{n,\gamma}))(b_n-b_0)\approx H_n(\beta_{n,\gamma}-\beta_{0,\gamma}).
\end{equation*}


By Lemma \ref{lem:limDerivsMn}, we have that $H_n = H_0+o_P(1)$ where $H_0$ is the Hessian of $M_0.$
Note that also by Lemma \ref{lem:limDerivsMn}, $X_n= X_0 + o_P(1),$ where $X_0$ is the cross derivative of the true objective $M_0.$  
Thus we have that 
\[-(J_n + (X_0+o_P(1))(b_n-b_0)) \approx (H_0+o_P(1))(\beta_{n,\gamma}-\beta_{0,\gamma}).\]
Rearranging the terms involving $o_P(1),$ and noting that,  by Assumption  \ref{as:bconsist}, $o_P(1)(b_n-b_0)=o_P(1)o_P(1)=o_P(1),$ we obtain
\[-(J_n + X_0(b_n-b_0)) \approx (H_0+o_P(1))(\beta_{n,\gamma}-\beta_{0,\gamma}) + o_P(1)(b_n-b_0).\]
By Assumption \ref{as:hessianInvert}, the Hessian term becomes invertible in the limit at $H_0$, and we multiply through and take transposes to get
\begin{align*}
-(H_0+o_P(1))^{-1}(J_n +  X_0(b_n-b_0)) \approx (\beta_{n,\gamma}-\beta_{0,\gamma}) + (H_0+o_P(1))^{-1}o_P(1)(b_n-b_0).
\end{align*}
Now, multiply both sides by $\sqrt{n}$ to obtain  
\begin{align}
\label{eq:toShowNormal}
\notag &\sqrt{n}(\beta_{n,\gamma}-\beta_{0,\gamma})\\
&\notag \approx -(H_0+o_P(1))^{-1}(J_n + X_0\sqrt{n}(b_n-b_0)) + (H_0+o_P(1))^{-1} o_P(1)\sqrt{n}(b_n-b_0)\\
&\notag =-(H_0+o_P(1))^{-1}\sqrt{n}(J_n +  X_0(b_n-b_0)) + (H_0+o_P(1))^{-1}O_P(1)o_P(1)\\
& = -(H_0+o_P(1))^{-1}(\sqrt{n} J_n +  X_0\sqrt{n}(b_n-b_0)) + o_P(1),
\end{align}
where $\sqrt{n}(b_n-b_0)=O_P(1)$ by Assumption \ref{as:bconsist} and the fact that $H_0+o_P(1)$ is invertible by Assumption \ref{as:hessianInvert} implies that its inverse is bounded in the limit, which allowed us to write that \[(H_0+o_P(1))^{-1}\sqrt{n}(b_n-b_0)^To_P(1)=(H_0+o_P(1))^{-1}O_P(1)o_P(1)=o_P(1).\]
To show that (\ref{eq:toShowNormal}) is asymptotically normal, we must show that $\sqrt{n}J_n$ is asymptotically normal.
Recall the definition of $J_n=E_n z$ from (\ref{eq:zn}). Note that
\begin{align}
\label{eq:rootnzn}
\sqrt{n} J_n&=\frac{\sqrt{n}E_n \fdr{\beta}{v} }{E_n r}
- \frac{\sqrt{n}E_n v (E_n \fdr{\beta}{r})}{(E_n r)^2}
- \sqrt{n}E_n\fdr{\beta}{\gamma  kl(\beta,b)}\\ 
&\notag \gol  
\frac{z_1}{E r}
- \frac{z_2 (E \fdr{\beta}{r})}{(E r)^2}
- z_3,
\end{align}
where $z_1,z_2,$ and $z_3$ are normally distributed by the Central Limit Theorem. By Slutsky's theorem and Remark \ref{rem:Er1}, $\sqrt{n}J_n$ is therefore asymptotically normal.
By the discussion in Appendix \ref{nuisinfluencefunction}, $\sqrt{n}(b_n-b_0)$ is asymptotically normal.
A final application of Slutsky's theorem to the expression in (\ref{eq:toShowNormal}) proves the result.
}

\subsection{Adaptive Lasso asymptotic normality}
Having shown that $\beta_{n,\gamma},$ the maximizer of ${M}_n,$ is well behaved in the limit, we will now prove a similar result for $\beta_{n,\gamma,\lambda},$ the maximizer of the full, double-penalized, relative sparsity objective, $W_n,$ which is defined in (\ref{eq:Wn}).
We prove this result by extending a result from \cite{zou2006adaptive}.
\newcommand{\thmadapt}{
Assume that, for $\delta>0,$ \[\frac{\lambda_n}{\sqrt{n}}\rightarrow 0 \text{ and }\lambda_n n^{(\delta-1)/2}\rightarrow \infty.\] Note that these conditions are necessary to obtain appropriate limiting behavior (either disappearance or predominance) of the penalty term.
 Define the active set, $\AC$,  to be the indices of the parameters that truly differ from their behavioral counterparts, 
 \begin{equation}
 \label{eq:AC}
 \AC=\{k:{(\beta_{0,\gamma,k}}-b_{0,k})\neq 0\}.
 \end{equation}
 Note that $\AC$ depends on $\gamma,$ since $\gamma$ defines the estimand without the relative sparsity penalty.
Define $\beta_{0,\gamma,\AC}$ as the coefficients indexed by $\AC.$

Then for the coefficients that differ from their behavioral counterparts, \[\sqrt{n}(\beta_{n,\gamma,\AC}-\beta_{0,\gamma,\AC})\gol N(0,(H_{0,\AC\AC})^{-1}\text{var}(r_{v_0}^T)((H_{0,\AC\AC})^{-1})^T).\]
where $r_{v_0}=[z_{0,\AC}^T+u_{0,\AC^C}^TH_{0,\AC^C\AC}+u_{0,\AC^C}^T{(}H_{0,\AC\AC^C}{)^T}+v_{0,\AC}^T(X_0^{T})_{\AC\AC}+v_{0,\AC^C}^T(X_0^{T})_{\AC^C\AC}]^T$ is a normally distributed combination of $z_0,H_0,X_0$ (which were defined in (\ref{eq:derivsDefn})), $u_0,$ and $v_0,$ and the subscripts $\AC,\AC^C$ indicate the indices associated with,  respectively, the non-behavioral and behavioral components of each vector or matrix. The coefficients that do not differ from their behavioral counterparts,  $\beta_{0,\gamma,\AC^C}$, converge to the limiting random variables of the behavioral policy estimators.
}

\begin{theorem}
\label{thmNormalityBetanLambda}
Asymptotic normality of $\beta_{n,\gamma,\lambda}.$ 
\thmadapt
\end{theorem}

\begin{proof}
See Appendix \ref{app:adaprf}.
\end{proof}



\newcommand{\adaprf}{
In Theorem \ref{thm:betanConsistNorm} (ii), we followed the standard proof of normality for two-step M-estimators that fulfill the classical conditions.  For this double penalized estimator ($\beta_{n,\gamma,\lambda}$) defined in (\ref{eq:betangammalambda}) proof, we instead follow   \cite{zou2006adaptive}.
Define, for ${M}_n$ in (\ref{eq:mn}), \[\Psi_n(\beta,b)=n {M}_n-\lambda_n \sum_{k=1}^K{w_{n,k}}|\beta_k-b_k|\]
and let
\begin{align*}
\Gamma_n(\beta)&=\Psi_n(\beta,b)-\Psi_n(\beta_{0,\gamma},b_0).
\end{align*}
Note  that, similarly to Theorem \ref{thm:betanConsistNorm} (ii), we can write 
\begin{align*}
\Psi_n(\beta,b)-\Psi_n(\beta_{0,\gamma},b_0)&
\approx 
    \sqrt{n}J_n\sqrt{n} (\beta-\beta_{0,\gamma})
    +\sqrt{n}z_{b,n}\sqrt{n}(b-b_0)\\
    &+ \frac{1}{2}\sqrt{n}(\beta-\beta_{0,\gamma})^T H_n\sqrt{n}(\beta-\beta_{0,\gamma})\\
      &
      +\frac{1}{2}\sqrt{n}(b-b_0)^T H_{bb,n} \sqrt{n}(b-b_0)\\
      &+ \sqrt{n}(b-b_0)^T X_{n}^{T}  (\beta-\beta_{0,\gamma}) \\
      &- \lambda_n\sum_{k=1}^K w_{n,k}\{|\beta_k-b_k|-|\beta_{0,\gamma,k}-b_{0,k}|\},
\end{align*}
where \[z_{b,n}=\fd{b}{{M}_n}{\substack{\beta=\beta_{0,\gamma}\\b=b_0}}\] and \[H_{bb,n}=\sds{b}{{M}_n}{\substack{\beta=\beta_{0,\gamma}\\b=b_0}}.\]
Both $\sqrt{n}J_n$ and $\sqrt{n}z_{b,n}$ are $O_P(1)$ by the Central Limit Theorem.

 Defining $u=\sqrt{n}(\beta-\beta_{0,\gamma})$  and $v=\sqrt{n}(b-b_0)$, let us now rewrite the expansion above following the style in \cite{zou2006adaptive} by substituting in $u,v$.
\begin{align*}
n{M}_n(u,v)-n{M}_n(0,0)&\approx \sqrt{n}J_{n} u + \sqrt{n}z_{b,n} v + u^T H_n u + v^T H_{b b, n}v + v^T X_n^{T}u
\end{align*}
Note that the terms that depend only on $v$ (and thus only on $b$) are constant with respect to $u$ (and by extension with respect to $\beta$) and, therefore, we do not need to consider them further in the objective function, which we will be maximizing with respect to $\beta$ only. Therefore, let us focus only on the redefined objective
\[\tilde{\Gamma}_n\approx \sqrt{n}J_nu + u^T H_n u + v^T X_n^{T} u,\] which,  as discussed in Theorem \ref{thm:betanConsistNorm} (ii), converges in probability to 
\[\tilde{\Gamma}_0\approx z_0{^T}u_0 + u_0^T H_0 u_0 + v_0^T X_0^{T} u_0.\] 

We can partition $\tilde{\Gamma}_0$ such that $u_{0,1}=u_{0,\mathcal{A}}, u_{0,2}=u_{0,\mathcal{A}^c}, v_{0,1}=v_{0,\mathcal{A}},v_{0,2}=v_{0,\mathcal{A}^c},$ etc, and rewrite,
 \begin{align}
 \label{eq:parti}
   \notag \tilde{\Gamma}_0\notag \approx \left[ {\begin{array}{cc}
   z_{0,1} & z_{0,2} \\
  \end{array} } \right]
  \left[ {\begin{array}{cc}
   u_{0,1}  \\
   u_{0,2}  \\
  \end{array} } \right]
  &+  \left[ {\begin{array}{cc}
   u_{0,1} & u_{0,2} \\
  \end{array} } \right]  \left[ {\begin{array}{cc}
   {(}H_0{)}_{11} & {(}H_0{)}_{12} \\
   {(}H_0{)}_{21} & {(}H_0{)}_{22} \\
  \end{array} } \right]  \left[ {\begin{array}{cc}
   u_{0,1}  \\
   u_{0,2}  \\
  \end{array} } \right] \\
  & +   \left[ {\begin{array}{cc}
   v_{0,1} & v_{0,2} \\
  \end{array} } \right]  \left[ {\begin{array}{cc}
   {(}X_0^{T}{)}_{11} & {(}X_0^{{T}}{)}_{12} \\
   {(}X_0^{{T}}{)}_{21} & {(}X_0^{{T}}{)}_{22} \\
  \end{array} } \right]  \left[ {\begin{array}{cc}
   u_{0,1}  \\
   u_{0,2}  \\
  \end{array} } \right].
\end{align}

This partitioning, as in  \citep{zou2006adaptive}, helps us pinpoint the parts of the objectives that are relevant to the  non-behavioral coefficients.  As we will see now, the penalty will go to zero for the pieces of the objective that depend on the non-behavioral  parameters and to infinity for the non-behavioral parameters. Note that a major innovation of \cite{zou2006adaptive} is to add dimension-specific weights, $w_{n,j}$, to the penalty term, which ensure this type of limiting behavior.

Let us now consider the penalty. We follow \cite{zou2006adaptive} closely here. 
Note that, by the continuous mapping theorem, Assumption \ref{as:bconsist}, and Theorem \ref{thm:betanConsistNorm} (i), 
\[w_{n,j}=
\frac{1}{|{\beta_{n,\gamma,j}}-b_{n,j}|^{\delta}}\gop \frac{1}{|{\beta_{0,\gamma,j}}-b_{0,j}|^{\delta}}.\]
Our relative sparsity penalty is
\begin{align*}
 &\frac{\lambda_n}{\sqrt{n}}\sum_{k=1}^K w_{n,k}\sqrt{n}\{|\beta_k-b_k|-|\beta_{0,\gamma,j}-b_{0,j}|\}\\
 &=  \frac{\lambda_n}{\sqrt{n}}\sum_{k=1}^K w_{n,k}\sqrt{n}\{|\beta_{0,\gamma,k}+\frac{u_k}{\sqrt{n}}-(b_{0,k}+\frac{v_k}{\sqrt{n}})|-|\beta_{0,\gamma,k}-b_{0,k}|\}\\
 &=
  \frac{\lambda_n}{\sqrt{n}}\sum_{k=1}^K w_{n,k}\sqrt{n}\{|(\beta_{0,\gamma,k}-b_{0,k})+\frac{u_k-v_k}{\sqrt{n}}|-|\beta_{0,\gamma,k}-b_{0,k}|\}
 .
\end{align*}
 We will now use the results in  \citep{zou2006adaptive} directly on the penalty term, but we will treat our shifted coefficient $(\beta-b)$ as they treat their raw coefficient $\beta$. 
Hence, following  \citep{zou2006adaptive}, if $(\beta_{0,\gamma,j}-b_{0j})\neq 0,$ then 
$w_{n,j}\gop |\beta_{0,\gamma,j}-b_{0,j}|^{-\delta}$ and
\begin{align*}
\sqrt{n}\{|(\beta_{0,\gamma,j}-b_{0,j})+\frac{u_j-v_j}{\sqrt{n}}|-|\beta_{0,\gamma,j}-b_{0,j}|\}\gop (u_{0,j}-v_{0,j})sgn(\beta_{0,\gamma,j}-b_{0,j}).
\end{align*}
By Slutsky's theorem, since $\lambda_n/\sqrt{n}\rightarrow 0,$ \[
\frac{\lambda_n}{\sqrt{n}}w_{n,j}\sqrt{n}\{|(\beta_{0,\gamma,j}-b_{0,j})+\frac{u_j-v_j}{\sqrt{n}}|-|\beta_{0,\gamma,j}-b_{0,j}|\}\gop 0.\]
If, on the other hand, $(\beta_{0,\gamma,j}-b_{0,j})=0,$ then
\begin{align*}
\sqrt{n}\{|(\beta_{0,\gamma,j}-b_{0,j})+\frac{u_j-v_j}{\sqrt{n}}|-|\beta_{0,\gamma,j}-b_{0,j}|= |u_j-v_j|.
\end{align*}
Since in this case \[\frac{\lambda_n}{\sqrt{n}}w_{n,j}=\frac{\lambda_n}{\sqrt{n}}n^{\delta/2}(\sqrt{n}(\beta_{n,\gamma}-b_n))^{-\delta},\] where
\[\sqrt{n}(\beta_{n,\gamma,j}-b_{n,j})=O_P(1)\] since  \begin{align*}\sqrt{n}(\beta_{n,\gamma,j}-b_{n,j})&=\sqrt{n}(\beta_{n,\gamma,j}-\beta_{0,\gamma,j}+\beta_{0,\gamma,j}-b_{n,j}-b_{0,j} + b_{0,j})\\&=\sqrt{n}(\beta_{n,\gamma,j}-\beta_{0,\gamma,j}-(b_{n,j}-b_{0,j})+\beta_{0,\gamma,j}-b_{0,j})\\&=\sqrt{n}(\beta_{n,\gamma,j}-\beta_{0,\gamma,j})-\sqrt{n}(b_{n,j}-b_{0,j})\\&=O_P(1)-O_P(1)=O_P(1).\end{align*}
We then have, since $\lambda_n/\sqrt{n}n^{\delta/2}\rightarrow \infty,$ that, by Slutsky's theorem,
\[\frac{\lambda_n}{\sqrt{n}}w_{n,j}\gop \infty.\]
We thus obtain complete penalization to the behavioral coefficient in this case.
Recall the previous partition of $\tilde{\Gamma},$ from (\ref{eq:parti}).
Coupled with the assumption that $\lambda_n n^{(\delta-1)/2}\rightarrow \infty,$ we have that
\[
\tilde{\Gamma}_n\gop \tilde{\Gamma}_0 = \begin{cases} z_0^T u_0 + u_0^T H_0 u_0{/2} + v_0^T {(}X_0^{{T}}{)} u_0 &\mbox{if } u_{0,j}-v_{0,j} =0, \forall j\not\in \mathcal{A} \\ 
-\infty &  otherwise \end{cases}.
\]
We see that even if the condition $u_{0,j} =v_{0,j}\forall j\not\in \mathcal{A}$ is satisfied, we do not have simplification of the objective, as is the case in \cite{zou2006adaptive}, since the terms $u_0$ and $v_0$ are not nonzero.
For the first case, we separate $u_{0,1}$ from $u_{0,2}$ to obtain 
\begin{align*}
z_0^Tu_0+u_0^TH_0u_0{/2}+v_0^T{(}X_0^{{T}}{)}u_0&= [z_{0,1}^T+u_{0,2}^T{(H_0)_{21}}{/2}+u_{0,2}^T{((H_0)_{12})^T}{/2}+v_{0,1}^T{(}X_0^{{T}}{)}_{11}+v_{0,2}^T{(}X_{0}^{{T}}{)}_{21}]^Tu_{0,1}\\&\ \ \ \ +u_{0,1}^T{(H_0)_{11}}u_{0,1}{/2}\\&=r_{v_0}u_{0,1}+u_{0,1}^T{(H_0)_{11}}u_{0,1}{/2},
\end{align*}
where
\[r_{v_0}=[z_{0,1}^T+u_{0,2}^T{(H_0)_{21}}{/2}+u_{0,2}^T({(H_0)_{12})^T}{/2}+v_{0,1}^T{(}X_0^{{T}}{)}_{11}+v_{0,2}^T{(}X_0^{{T}}{)}_{21}].\]
We then differentiate with respect to $u_{0,1}$ to see that the penalized objective, $W_0,$ is maximized at
\[u_{0,\AC}=-{((H_0)_{11})^{-1}}r_{v_0}^{{T}}.\]
If, on the other hand, $\beta_{0,\gamma,j}-b_{0,j}\neq 0,$ which occurs for $ j\not\in \AC,$ then the weight of the penalty goes off to negative infinity.  In this case, we must minimize $|u_j-v_j|,$ which drives $u_j$ to $v_j$.   
Finally, if we let
$u_n=\arg\max_{u} \tilde{\Gamma}_n$ and $u_0=\arg\max_{u} \tilde{\Gamma}_0,$ then, since
$\tilde{\Gamma}_n\gop \tilde{\Gamma}_0,$ we have that, by epi-convergence  \citep{geyer1994asymptotics,knight2000asymptotics},
$u_n\rightarrow u_0,$ where
\[u_0 = \left[u_{0,{\AC}},u_{0,\AC^c}\right]^T=\left[-{((H_0)_{11})^{-1}}r_{v_0}^{{T}},v_{\AC^c}\right]^T.\]
Thus we have proved the theorem.


}
Thus we obtain, for the truly non-behavior coefficients, which are the solutions to the maximization of the base objective, ${M}_n,$ asymptotic normality, because these coefficients are untouched (in the limit) by the adaptive penalty. We simultaneously drive the truly behavioral coefficients to their behavioral counterparts, which are also asymptotically normal, because they maximize  likelihood. 
\section{Inference for policy coefficients}

\subsection{Estimator for the variance of the coefficients}

\label{avar}
We have given expressions for the asymptotic forms of the coefficients $\beta_{n,\gamma}$ and $\beta_{n,\gamma,\lambda}$ based on the theory above, but now we give more detail on how to estimate the variances of these estimators in practice.  To do so, we will need to build on the partial derivatives of ${M}_n,$ the forms of  which are given in Lemma \ref{lem:limDerivsMn} and Appendix \ref{app:gradients}.

\newcommand{\grads}{
We will provide forms for the derivatives below. Note that $s,\beta,b\in \mathbb{R}^K$, $a\in\{0,1\},$
and $\pi(A=1|s)\in\mathbb{R}.$
For post-selection inference, we must define also the nuisance $b$ and $1_{\mathcal{A}}=(I(1\in \mathcal{A}),\dots,I(K\in \mathcal{A}))^T$ where $\mathcal{A}$ is the active set.  We then obtain policies
\begin{align*}
&\pi_{\beta,b}(A=a|s)=\pi_{\beta,b}(A=1|s) (1-\pi_{\beta,b}(A=1|s))^{1-a},\\
&\pi_{\beta,b}(A=1|s)=\expit(\beta^T(s\odot 1_{\mathcal{A}})+b^T(s\odot (1-1_{\mathcal{A}})).
\end{align*}
Note that the derivative of $\pi_{\beta,b}(A=1|s)$ with respect to its linear argument is  \[d\pi_{\beta,b}=\pi_{\beta,b}(1|s)(1-\pi_{\beta,b}(1|s)),\] which will occur below in numerous places.
For post-selection inference, we have 
\begin{align*}
&\fdr{\beta}{\log \pi_{\beta,b}(a|s)}=(a-\pi_{\beta,b}(1|s))(s\odot 1_{\mathcal{A}}),\\
&\fdr{b}{\log \pi_{\beta,b}(a|s)}=(a-\pi_{\beta,b}(1|s))(s\odot (1-1_{\mathcal{A}})),\\
&\sdsr{\beta}{\log \pi_{\beta,b}(a|s)}=-d\pi_{\beta,b}(s\odot 1_{\mathcal{A}})(s\odot 1_{\mathcal{A}})^T,\\
&\sddr{b}{\beta}{\log\pi_{\beta,b}(a|s)}=- d\pi_{\beta,b}(s\odot 1_{\AC}) (s\odot (1-1_{\AC}))^T. 
\end{align*}
Recall as well that for active set $\mathcal{A},$ $1_{\mathcal{A}}=(1_{1\in \mathcal{A}},\dots,1_{K \in \mathcal{A}})^T.$
We can use properties of the derivative of the log function, which is a popular trick in the policy gradient literature  \citep{sutton2018reinforcement}:
\begin{align*}
&\fdr{\beta}{\prod_{t=0}^T \pi_{\beta,b}(A_t|S_t)}=\prod_{t=0}^T \pi_{\beta,b}(A_t|S_t)\fdr{\beta}{\log \prod_{t=0}^T \pi_{\beta,b}(A_t|S_t)},\\
&\fdr{\beta}{\log \prod_{t=0}^T \pi_{\beta,b}(A_t|S_t) }={\sum_{t=0}^T\fdr{\beta} \log \pi_{\beta,b}(a_t|s_t)},\\
&\fdr{b}{\log \prod_{t=0}^T \pi_{\beta,b}(A_t|S_t) }={\sum_{t=0}^T\fdr{b} \log \pi_{\beta,b}(a_t|s_t)},\\
&\sddr{b}{\beta}{\log \prod_{t=0}^T \pi_{\beta,b}(A_t|S_t) }={\sum_{t=0}^T\sddr{b}{\beta} \log \pi_{\beta,b}(a_t|s_t)},\\
&\sdsr{\beta}{\log \prod_{t=0}^T \pi_{\beta,b}(A_t|S_t) }={\sum_{t=0}^T\sdsr{\beta} \log \pi_{\beta,b}(a_t|s_t)}.
\end{align*}

For the behavioral policy, recall that we have the following likelihood for one trajectory
\begin{align*}
&l_T(b)=\log \prod_{t=0}^T \pi_b(a_t|s_t).
\end{align*}
We then have score and hessian of the behavioral likelihood, which can be written as functions of the  derivatives defined above, as
\begin{align*}
&\fdr{b}{l_T(b)}=\sum_{t=0}^T \fdr{b}{\log \pi_b(a_t|s_t)},\\
&\sdsr{b}{l_T(b)}=\sum_{t=0}^T \sdsr{b}{\log \pi_b(a_t|s_t)}.
\end{align*}

}

\newcommand{\varDeriv}{





Let us now derive an estimator for the asymptotic variance of $\beta_{n,\gamma}.$ 
Recall that $z_0,$ $H_0,$ and $X_0$ are the expected gradient with respect to $\beta$, the expected Hessian with respect to $\beta,$ and the expected cross derivative with respect to $b$ and $\beta,$ all of which are defined in (\ref{eq:derivsDefn}). 
Recall that in Theorem \ref{thm:betanConsistNorm} (ii) we used a Taylor expansion in $\beta,b$ to obtain
\[\sqrt{n}(\beta_{n,\gamma}-\beta_{0,\gamma})\approx-H_n^{-1}(\sqrt{n}J_n + X_n\sqrt{n}q_n ),\] where $J_n$ and $X_n$ are defined in (\ref{eq:derivsDefn}), along with their expectations, and $q_n$ is defined in (\ref{eq:q}).
Recall that 
 \[r = \frac{\prod_{t=0}^T \pi_{\beta,b}(A_t|S_t)}{\prod_{t=0}^T \pi_{b}(A_t|S_t)}, v = r\sum_{t=0}^{T}R(S_t,A_t,S_{t+1}).\] 
Recall that by the Law of Large Numbers, Slutsky's theorem, and Remark \ref{rem:Er1}, that $H_n \gop H_0.$  
Recall also the definition of the cross derivative, $X_n,$ in Equation (\ref{eq:Xn}), and recall that by 
the Law of Large Numbers, Slutsky's theorem, and Remark \ref{rem:Er1}, $X_n\gop X_0.$  

We finally  consider $J_n=E_n z,$ where $J_n$ is also defined in  (\ref{eq:derivsDefn}), and where
recall that, as defined in (\ref{eq:z}),
\begin{equation}
z = \fdr{\beta}{(v/E_nr - \gamma kl)}. 
\end{equation}
Recall also that we have that 
\begin{align*}
\sqrt{n}  J_n=\sqrt{n}E_n z &=\frac{\sqrt{n}E_n \fdr{\beta}{v} }{E_n r}
- \frac{\sqrt{n}E_n v (E_n \fdr{\beta}{r})}{(E_n r)^2}
- \sqrt{n}E_n\fdr{\beta}{\gamma  kl(\beta,b)}\\ 
&\gol  
\frac{z_1}{E r}
- \frac{z_2 (E \fdr{\beta}{r})}{(E r)^2}
- z_3,
\end{align*}
where $z_1,z_2,z_3$ are normally distributed random variables by the CLT, Slutsky's theorem, and Remark \ref{rem:Er1}, and hence their scaled sum is also normally distributed. Set their scaled sum to be $z_0=\frac{z_1}{E r}
- \frac{z_2 (E \fdr{\beta}{r})}{(E r)^2}
- z_3.$
Replace the empirical distributions involving $r$ in $z$ to define
\begin{align}
\label{eq:tildez}
\tilde{z}
=
\frac{ \fdr{\beta}{v} }{E r}
- \frac{v (E \fdr{\beta}{r})}{(E r)^2} 
- \fdr{\beta}{\gamma  kl(\beta,b)}.
\end{align}
Note that still
\begin{align*}
\sqrt{n} E_n \tilde{z}&=\frac{\sqrt{n}E_n \fdr{\beta}{v} }{E r}
- \frac{\sqrt{n}E_n v (E \fdr{\beta}{r})}{(E r)^2}
- \sqrt{n}E_n\fdr{\beta}{\gamma  kl(\beta,b)}\\ 
&\gol  
\frac{z_1}{E r}
- \frac{z_2 (E \fdr{\beta}{r})}{(E r)^2}
- z_3 =z_0
\end{align*}
We have then that, showing explicitly  that the derivatives are evaluated at their true parameters,
\begin{align*}
\sqrt{n}(\beta_{n,\gamma}-\beta_{0,\gamma}) &\approx - \sqrt{n} E_n ((H_n(\beta_{0,\gamma},b_0))^{-1} (z(\beta_{0,\gamma},b_0)+X_n(\beta_{0,\gamma},b_0)q(b_0)))\\&\gol  ((H_0(\beta_{0,\gamma},b_0))^{-1} (z_0(\beta_{0,\gamma},b_0)+X_0(\beta_{0,\gamma},b_0)q_0(b_0) )),
\end{align*}
where $z_0$ and $q_0$ are normally distributed random variables and $H_0$ and $X_0$ are constants.  Note that, now 
\begin{align*}
&var((H_0(\beta_{0,\gamma},b_0))^{-1} (z_0(\beta_{0,\gamma},b_0)+X_0(\beta_{0,\gamma},b_0)q_0(b_0) ))\\
&=(H_0(\beta_{0,\gamma},b_0))^{-1} var(z_0(\beta_{0,\gamma},b_0)+X_0(\beta_{0,\gamma},b_0)q_0(b_0)) ((H_0(\beta_{0,\gamma},b_0)^{-1})^T)
\end{align*}
Now, plug in for $z_0$ and $q_0$ with $\sqrt{n}E_n \tilde{z}$ and $\sqrt{n}E_n q$ to get
\begin{align}
\label{eq:removedWeight}
\notag &(H_0(\beta_{0,\gamma},b_0))^{-1} var(\sqrt{n}E_n \tilde{z}(\beta_{0,\gamma},b_0)+X_0(\beta_{0,\gamma},b_0)\sqrt{n}E_n q(b_0)) ((H_0(\beta_{0,\gamma},b_0)^{-1})^T)\\
 &=(H_0(\beta_{0,\gamma},b_0))^{-1} var(\sqrt{n}E_n (\tilde{z}(\beta_{0,\gamma},b_0)+X_0(\beta_{0,\gamma},b_0) q(b_0))) ((H_0(\beta_{0,\gamma},b_0)^{-1})^T)\\
\notag &=(H_0(\beta_{0,\gamma},b_0))^{-1} var( \tilde{z}(\beta_{0,\gamma},b_0)+X_0(\beta_{0,\gamma},b_0) q(b_0)) ((H_0(\beta_{0,\gamma},b_0)^{-1})^T)\\
\notag &=(H_0(\beta_{0,\gamma},b_0))^{-1} E(( \tilde{z}(\beta_{0,\gamma},b_0)+X_0(\beta_{0,\gamma},b_0) q(b_0))^2) ((H_0(\beta_{0,\gamma},b_0)^{-1})^T)
\end{align}
where the variance turns into the expectation of the square because the true derivative evaluated at the true parameters is zero.  
We can replace the true expectation, $E,$ with its empirical version, $E_n:$
\[(H_0(\beta_{0,\gamma},b_0))^{-1} E_n(( \tilde{z}(\beta_{0,\gamma},b_0)+X_0(\beta_{0,\gamma},b_0) q(b_0))^2) ((H_0(\beta_{0,\gamma},b_0)^{-1})^T).\]
Note that $\tilde{z}$, defined in (\ref{eq:tildez}), depends on $Er$. Hence, we can plug in for $Er,$ and hence replace 
$\tilde{z}$ with $z,$ where recall that, as defined in (\ref{eq:z}),
\begin{equation}
z = \fdr{\beta}{(v/E_nr - \gamma kl)}. 
\end{equation}
This allows us to then reintroduce the weighting term, even though we had to remove it earlier in (\ref{eq:removedWeight}) to perform manipulations of the variance operator, which requires independence of the summands, which is broken by the dependence of  $E_nr$ on the joint distribution of the $n$ samples. Hence, replacing $\tilde{z}$ with $z,$ we obtain,
\[(H_0(\beta_{0,\gamma},b_0))^{-1} E_n(( z(\beta_{0,\gamma},b_0)+X_0(\beta_{0,\gamma},b_0) q(b_0))^2) ((H_0(\beta_{0,\gamma},b_0)^{-1})^T).\]
We finally plug in $\beta_{n,\gamma},b_n, H_n,$ and $X_n$  for $\beta_{0,\gamma},b_0, H_0,$ and $X_0,$ to get the final form of our estimator
\[\sigma^2_n = (H_n(\beta_{n,\gamma},b_n))^{-1} E_n(( z(\beta_{n,\gamma},b_n)+X_n(\beta_{n,\gamma},b_n) q(b_n))^2) ((H_n(\beta_{n,\gamma},b_n)^{-1})^T).
\]
}

Recall that $H_n,$ defined in (\ref{eq:derivsDefn}),  is an estimator for the Hessian $H_0.$ Recall that $X_n,$ defined in (\ref{eq:Xn}), is an estimator for the cross derivative, $X_0$.
Define the function over which we take an empirical expectation in the gradient of (\ref{eq:mn}) as
\begin{equation}
\label{eq:z}
z = \frac{\partial}{\partial \beta} \left (\frac{\frac{\prod_{t=0}^T \pi_{\beta}(A_t|S_t)}{\prod_{t=0}^T \pi_{b}(A_t|S_t)}\sum_{t=0}^TR(S_t,A_t,S_{t+1})}{E_n \frac{\prod_{t=0}^T \pi_{\beta}(A_t|S_t)}{\prod_{t=0}^T \pi_{b}(A_t|S_t)}}-\gamma \log\left(\frac{\prod_{t=0}^T \pi_{b}(A_t|S_t)}{\prod_{t=0}^T \pi_{\beta}(A_t|S_t)}\right)\right )\bigg|_{\substack{\beta=\beta_{0,\gamma}\\b=b_0}}. 
\end{equation}
Note that $J_n = E_n z,$ for $J_n$ defined in (\ref{eq:derivsDefn}). For $z$  defined in (\ref{eq:z}), define $z_i$ such that $J_n=\frac{1}{n}\sum_i z_i$. 
Define 
\begin{equation}
\label{eq:q}
q=(E_nl''(b_n))^{-1}l'(b_n) \text{ and } q_n = E_n q.
\end{equation}
For $q$ defined in (\ref{eq:q}), define $q_i$ such that $q_n = \frac{1}{n}\sum_i q_i$. 
 
 We then derive in Appendix \ref{app:avarbeta} the following estimator for $\sigma^2_0,$ the asymptotic variance of $\beta_{n,\gamma},$
\begin{equation}
\label{eq:sigma2n}
\sigma^2_n = (H_n)^{-1} \left(\frac{1}{n}\sum_i( z_{i}+X_n {q_{i}})^2\right) \left(H_n^{-1}\right)^T,
\end{equation}
where all terms are evaluated at plug-in estimates, $\beta_{n,\gamma}$ and $b_n$, of the true parameters, $\beta_{0,\gamma}$ and $b_0$, respectively. 

We use (\ref{eq:sigma2n}) to perform rigorous inference in the post-selection step.  We also use the estimator (\ref{eq:sigma2n}) to obtain a rough visualization of the variability of $\beta_{n,\gamma,\lambda}$ in the training data, which is heuristic, but can be helpful for assessing the impact of $\gamma$ on the stability of $M_n=V_n-\gamma KL_n$, defined in (\ref{eq:mn}). The stability of $M_n$ is important for inference in the post-selection step. To assess this stability, we can also examine the variance of the value, an estimator for which is given below.



\newcommand{\avarbeta}{
We can write
\begin{align*}\sqrt{n}(\beta_{n,\gamma}-\beta_{0,\gamma})&\approx-H^{-1}\left(\sqrt{n}E_n \frac{d}{d\beta}m + X(l'')^{-1}\sqrt{n}E_n l' \right)\\
&=-\sqrt{n}E_nH^{-1}\left(\frac{d}{d\beta}m + X(l'')^{-1}l' \right)\\
&=-\sqrt{n}E_n AB,
\end{align*}
where $A=H^{-1}$ and $B=\frac{d}{d\beta}m + X(l'')^{-1}l'.$
Then,
\begin{align}
\label{eq:sigbetangamma}
\sigma^2=var(\sqrt{n}(\beta_{n,\gamma}-\beta_{0,\gamma}))\approx var(AB)=Avar(B)A^T= AE[BB^T]{A}^T\approx AE_n[BB^T]{A}^T.
\end{align} The second to last equality follows because
\[var(B)=E(BB^T)-E(B)E(B)T\] and 
\[E(B)=E(z+X(l'')^{-1}l')=E(\beta-\beta_0)+X(l'')^{-1}El'=0,\] where the last statement follows because  $\frac{d}{db}l|_{b=b_0}.$
}

\subsection{Estimator for the variance of the value}
To select $\lambda,$ the tuning parameter for the relative sparsity penalty in (\ref{eq:Wn}),
it is useful to estimate the variance of the value function. We use the following plug-in estimator, which is derived in Appendix \ref{app:valuevar},
\begin{equation}
\label{valuevar}
\sigma^2_{n,V,w} = \frac{\frac{1}{n} \sum_{i=1}^n \left(\frac{\prod_{t=0}^T \pi_{\beta}(A_{i,t}|S_{i,t})}{\prod_{t=0}^T \pi_{b}(A_{i,t}|S_{i,t})}\sum_{t=0}^T R(S_{i,t},A_{i,t},S_{i,t+1}) - {V}_n\right)^2}{\left(\frac{1}{n}\sum_{i=1}^n \frac{\prod_{t=0}^T \pi_{\beta}(A_{i,t}|S_{i,t})}{\prod_{t=0}^T \pi_{b}(A_{i,t}|S_{i,t})}\right)^2}.
\end{equation}
This estimator will give us a sense of the uncertainty of the value function.
This estimator is especially important for the selection rule in (\ref{lambda.crit.diff}), where we require a standard error of the value under the behavioral policy.  In other words, we require the variance of $V_n(b).$ When $\beta=b,$ the estimator in  (\ref{valuevar}) reduces to a simple empirical variance.

\newcommand{\valuevar}{
Recall that by Slutsky's theorem  and Remark \ref{rem:Er1},
\begin{equation}
    V_n=\frac{E_n v}{E_n r}\gop \frac{E v}{E r}=E v.
\end{equation}
Now, because $E_n v$ is an average, ${\sqrt{n}}{E_n v} \gol Z_V \rightarrow N(E v, \sigma_V^2)$ for some $\sigma_V^2.$ We would like to estimate $\sigma_V^2.$  In the unweighted case, we can use the estimator derived in {\citerelativesparsity},
\begin{equation}
\label{eq:sigmaVnest}
    \sigma_{V,n}^2(\tilde{V}_n) = var(\sqrt{n}\tilde{V}_n)=\frac{1}{n} \sum_{i=1}^n \left(\frac{\prod_{t=0}^T \pi_{\beta}(A_{i,t}|S_{i,t})}{\prod_{t=0}^T \pi_{b}(A_{i,t}|S_{i,t})}\sum_{t=0}^T R(S_{i,t},A_{i,t},S_{i,t+1}) - \tilde{V}_n\right)^2.
\end{equation}
In the weighted case, we note that 
\begin{equation}
    \sqrt{n}V_n=\frac{\sqrt{n}E_n v}{E_n r}\gop \frac{Z_V}{E r}.
\end{equation}
Hence, we would like to estimate the variance of the weighted importance sampling estimator, 
\[\sigma^2_{0,V,w}=var\left(\frac{Z_V}{Er}\right)=\frac{1}{(E r)^2}var(Z_V)=\frac{1}{(E r)^2} \sigma_V^2,\] where we have used a similar strategy in terms of plugins as we did in step (\ref{eq:removedWeight}) of the derivation of the variance of $\beta_{n,\gamma}$, which is defined in (\ref{eq:sigma2n}). We plug in for $\sigma_V^2$ using (\ref{eq:sigmaVnest}) and we plug in for $E r$ with $E_n r$ to obtain a final estimator for the weighted objective function, 
\[
\sigma^2_{n,V,w} = \frac{\sigma_{V,n}^2(V_n)}{(E_n r)^2}.
\]
}

\newcommand{\scaling}{

}
\subsection{Estimation algorithm}
\begin{algorithm}[H]
\begin{algorithmic}
\STATE First, divide the dataset in half into split 1 (for selection) and split 2 (for post-selection inference).
\bindent
\STATE  Divide split 1 into a split 1 train and split 1 test.
\STATE  In split 1 train, optimize (\ref{eq:mn}) to obtain pilot estimators, $\beta_{n,\gamma},$ and maximize likelihood to obtain $b_n,$  Use these to optimize (\ref{eq:Wn}) to obtain $\beta_{n,\gamma,\lambda}$ for three values of $\gamma,$ three values of $\delta,$ and ten values of $\lambda.$ 
\STATE In split 1 train, visualize each $(\gamma,\delta)$ pair over varying $\lambda.$  
\STATE In split 1 train, choose $(\gamma,\delta)$ based on the visualizing the variability of the estimates $\beta_{n,\gamma,\lambda}$ using estimator (\ref{eq:sigma2n}).  
\STATE In split 1 train, simultaneously, view the value in (\ref{eq:vnestw}) as a function of $\gamma,\lambda,$ and $\delta,$ along with its variance, using the estimator (\ref{valuevar}). 
\STATE Also, using split 1 test, estimate the value. 
\STATE Given $\gamma$ and $\delta$, select $\lambda$ using (\ref{lambda.crit.diff}) in split 1 train. 
 \eindent
\STATE Given a $\gamma,\delta,$ and $\lambda,$ which creates selection $1_{\AC},$ where $\AC$ indexes the selected covariates: 
\bindent
\STATE In split 2 (the post-selection step), optimize (\ref{eq:mn}) to obtain $\beta_{n,\gamma,\AC}$ (where $\AC$ in $\beta_{n,\gamma,\AC}$ indexes the selected components) . 
\STATE In split 2, perform inference using Theorem \ref{thm:betanConsistNorm} and the corresponding estimator (\ref{eq:sigma2n}), which provides confidence intervals for the selected components, $\beta_{0,\gamma,\AC}$.
\eindent
\end{algorithmic}
\end{algorithm}

\section{Simulations}
We perform simulations to better understand the behavior of our estimators in Equations (\ref{eq:sigma2n}) and (\ref{valuevar}) as a function of the tuning parameters $\gamma$, $\lambda$, and $\delta$ (recall that these parameters are discussed in Section \ref{sec:estTuning}).

\subsection{Simulation scenario}
\label{sim.sett}
For our simulations, 
Set $\alpha=P(\text{type 1 error})=0.05$.  
We will fix the sample size to be $n=1000,$ the number of Monte-Carlo repetitions for the selection 
to be $M_S=100,$ the number of Monte-Carlo repetitions for coverage 
to be $M_C=500$,  the trajectory length to be $T=2,$ the state to be $S\in \mathbb{R}^K,$ where $K=2$, and the true behavioral policy parameter to be $b_0=(-0.3,0.2)^T$.  We further define the reward to be
\begin{equation}
\label{eq:rewardspec}
R(s_t,a_t,s_{t+1})=-s_{t,2}a_t.
\end{equation}  
The reward in (\ref{eq:rewardspec}) is convenient, because it illustrates our method and also allows us to derive theoretical results with respect to $\beta_{0,\gamma}.$ These results are discussed in Appendix \ref{supp.section.estimating.beta.gamma} and Appendix \ref{sim.coef} and can be used to evaluate coverage in the simulation studies.

\subsection{Data generation}

For our data generation, we require that the states have  constant variance over time, which is the case with many physiological measurements, and we therefore use the generative model in \cite{ertefaie2018constructing}, which we will describe briefly below.
We set  the standard deviation of $\epsilon,$ to the $K\times K$ identity matrix, $\sigma_{\epsilon}=I_K,$  and we set the treatment effect  to be $\tau_k=0.1$ for all $k$. Let  $\mu_0=1_K,$  where $1_K$ is a $K$-dimensional vector of ones, and draw the first state and action, $S_0\sim N(0_K,I_K)$ and $A_0|S_0 \sim Bern(\expit(b_0^TS_0)).$ For each ensuing time step $t$, for dimension $k$, draw ${\epsilon_k}\sim N(1,\sigma_{\epsilon,k}^2),$ where $\sigma_{\epsilon,k}^2=1$.  Then draw
\[S_{t,k}|S_{t-1,k},A_t = {(S_{t-1,k}-\mu_{t-1,k}+\epsilon_k})/{(1+\sigma^2_{\epsilon,k})^{1/2}}+\mu_{t,k},\] where  $\mu_{t+1,k} = \mu_{t,k}(1+\tau_k A_t),$ and
 $A_t|S_t \sim Bern(\expit(b_0^T S_t)).$
Note that $\sigma_{\epsilon,k}^2$ is constant over time, and that the division by $\sigma_{\epsilon,k}^2$ in the expression for $S_{t,k}$ ensures that the variance of the states are constant over time.
In simulations, because we set $\sigma^2_{\epsilon,k}=1,$ we generate states with unit variance.  We then do not have to scale the states to have equal variance as discussed in \ref{app:scaling}. This is useful because, as described in Section \ref{supp.section.estimating.beta.gamma}, we will be using the fact that we know the functional form of the reward in (\ref{eq:rewardspec})  to analytically derive an expression to estimate $\beta_{0,\gamma}$, but this expression will be impacted by scaling of the state.  When we just observe the reward , as we do in the real data, this is not an issue, and we do scale the states.

\newcommand{\IPTW}{
\[E(R(1))-E(R(0))\approx \frac{1}{n}\sum_{i=1}^n \frac{R_i A_i}{\pi_{b_n}(A_i=1|S_i)}-\frac{1}{n}\sum_{i=1}^n \frac{R_i (1-A_i)}{1-\pi_{b_n}(A_i=1|S_i)}\]
}

\subsection{Selection diagrams}
\label{sec:selectionDiagramsSim}
\begin{figure}[htp!]
\centering
\includegraphics[width=0.98\textwidth]{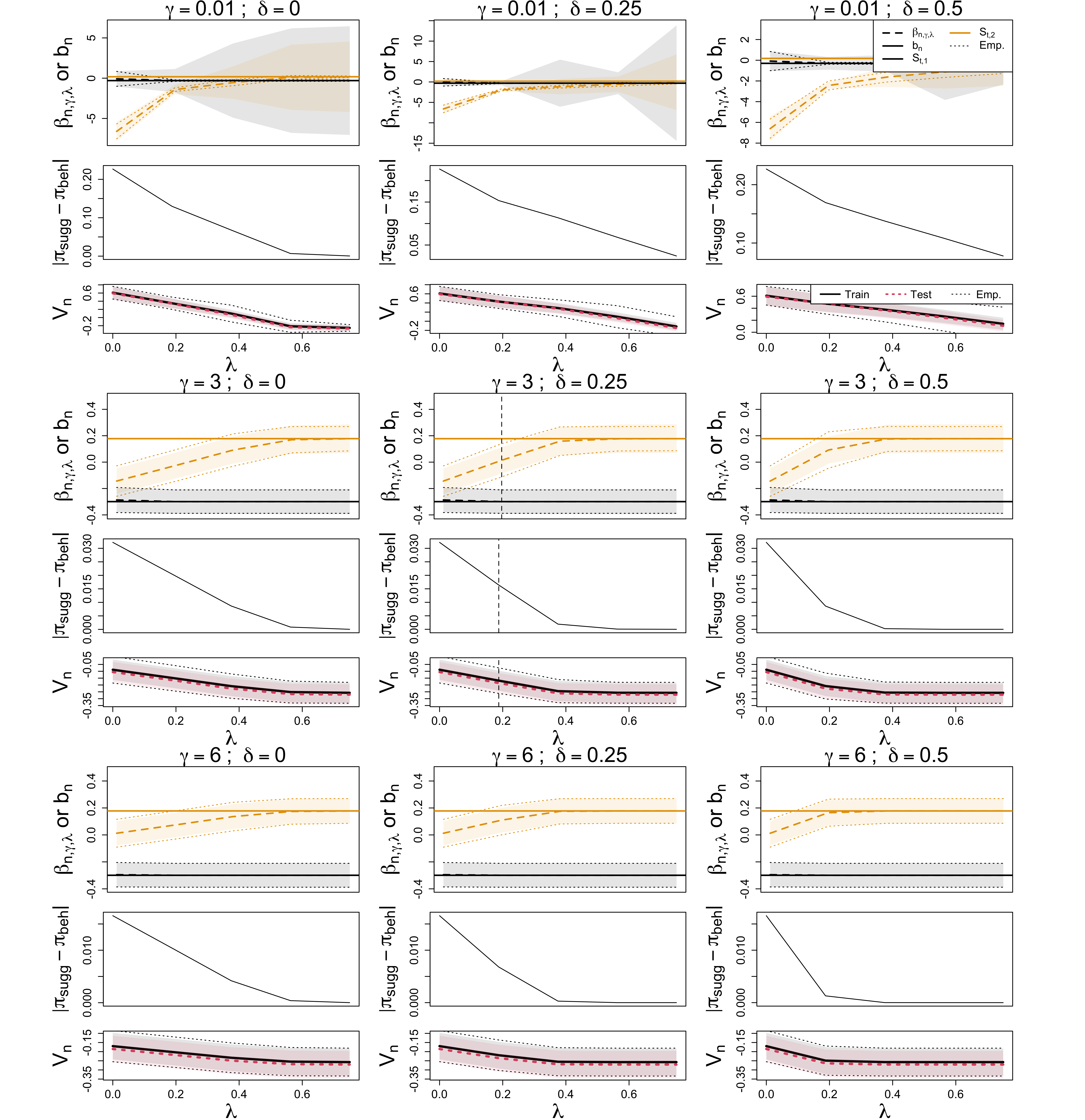}
	\caption[Selection diagrams for the simulated data]{
\textbf{Selection diagrams for the simulated data ($n_{train}=250,n_{test}=250$).}
	 Over increasing $\gamma$ (going down) and $\delta$ (going right), we show the average coefficients in the suggested ($\beta_{n,\gamma,\lambda})$ and behavioral ($b_n$) policies, the average difference in probability of treatment between the two policies, and the average value ($V_n$) for the suggested policy, all of which were computed in the first split of the data . 
The dotted vertical line indicates $\lambda_n$, a choice of $\lambda$  based on (\ref{lambda.crit.diff}). Note that the average suggested policy probability of treatment is $\pi_{sugg}=(1/nT)\sum_i\sum_t{\pi}_{\beta_{n,\gamma,\lambda}}(A_{i,t}=1|s_{i,t})$ and vice versa for $\pi_{beh}$.
The shaded regions in the coefficient ($\beta_{n,\gamma,\lambda}$) panels correspond to (\ref{eq:sigma2n}), and the dotted lines show one standard error estimated empirically.
The shaded regions in the value panels show one standard error based on (\ref{valuevar}), which was used to select $\lambda_n$  using (\ref{lambda.crit.diff}), and the dotted lines show one standard error estimated empirically. 
 }
\label{fig:sim}
\end{figure}
  Figure \ref{fig:sim} shows results for $M_S=100$ Monte-Carlo datasets, and panels are arranged in triplets indexed by $\gamma$ and $\delta$, where $\gamma$ and $\delta$ increase as we descend the plot or go to the right, respectively. 
We see that some variables approach their behavioral values less rapidly, giving us  relative sparsity, as observed in {\citerelativesparsity}.  Under the assumption that (\ref{eq:thetapolicyInf}) holds, and recalling that $\beta_0=\arg\max V_0,$ we expect to see results that align with the fact that the unconstrained maximizer is $\beta_{0}=(0,-\infty),$ for which a proof is provided in Appendix \ref{sim.coef}. Hence, we expect the sign of the coefficient $\beta_{n,2}$ to be negative and to become larger in magnitude as $\lambda\rightarrow 0,$ which aligns with the results in Figure \ref{fig:sim}. 

When $\gamma$ is large enough, and therefore the objective function is stable enough, we see that the ``empirical'' standard errors of the coefficients, which are shown as dotted black lines, align with variability shown by (\ref{eq:sigma2n}), which are shown as shaded regions.   

In the second panel of each triplet, which shows the difference in the probability of treatment under the behavioral and suggested policies, we see that increasing $\gamma$ gives us some baseline closeness to behavior, and the gap is further closed by increasing $\lambda$.  
The central triplet in Figure \ref{fig:sim} contains vertical dotted lines indicating the selected policy.   For each dataset, we automatically choose $\lambda_n,$ based on (\ref{lambda.crit.diff}), such that the corresponding suggested policy is as sparse as possible but has an increase in value of one standard error above the behavior policy. The dotted vertical  line in the central triplet indicates the choice of $\lambda_n$  when using the coefficients averaged over Monte-Carlo datasets.  
  Note that this selection, using (\ref{lambda.crit.diff}), was conducted only on one split of the data, which was itself split into a training set and a test set. The former is used to estimate coefficients and the latter to assess held-out value.  We see good overall closeness to behavior for the selected policy, and that the second covariate was selected, as expected.

In general, for $V_n$ in Figure \ref{fig:sim}, for large enough $\gamma$, we see that the shaded regions, which correspond to one standard error of $V_n$ according to the theoretical estimator in (\ref{valuevar}), are close to the dotted lines, which correspond to one standard error estimated empirically. 
For both the coefficients and the value,  the objective function becomes more unstable for small $\gamma.$ It is important not to attempt to attain too much value, $V_n$, by making $\gamma$ small, because the more one upweights $V_n,$  the more unstable the estimation and inference become.  However, Figure \ref{fig:sim} exaggerates this instability, because it is based on only half of the data, which is further divided in half (one half is used for training and one half to compute held out value). Since the post-selection inference will be conducted on a larger sample, and the number of covariates will decrease after selection, both of which stabilize the estimation and inference, as discussed in Section \ref{sec:estTuning}, one can select a $\gamma$ that is slightly smaller than what Figure \ref{fig:sim} might suggest.  
  Having made selections of $\gamma$ and $\lambda$,  we then perform inference on a held out, independent split of the dataset, for which we will now provide results in Section \ref{sec:postselect}.

\subsection{Post-selection inference}
\label{sec:postselect}
For each dataset, to avoid issues with post-selection inference that occur when one selects and does inference on the same dataset \citep{potscher2009distribution}, we first perform selection using (\ref{eq:Wn}) in one split of the data, and then we  use a second, independent split for inference. For the latter, we perform inference on the non-behavioral components of $\beta_{n,\gamma}$ within the suggested policy $\pi_{\beta,b},$ which is defined in Section \ref{sam.split}, Equation (\ref{eq:postSelectPolicy}). Recall that  each non-selected (behavioral) coefficient, $\beta_{k}$ in $\pi_{\beta,b},$ which is defined in (\ref{eq:postSelectPolicy}),  is fixed in advance to its behavioral counterpart, $b_{n,k}$. We therefore re-estimate and perform inference on the non-behavioral components of $\beta_{\gamma,n}$.  We show post-selection inference results in Table \ref{tab.post.inf}.

\begin{table}[ht]
\centering
\begin{tabular}{lll}
 & Suggested ($\beta_{n,\gamma})$ & Behavioral ($b_n$) \\ 
  \hline
$S_{t,1}$ & set to $b_n$ & -0.302 (-0.428, -0.176) \\ 
  $S_{t,2}$ & -0.129 (-0.261, 0.003) & 0.189 (0.066, 0.312) \\ 
\end{tabular}
\caption[Simulation post-selection inference]{Estimated coefficients (95\% confidence interval) for held out dataset post-selection inference; $n=$500, $\gamma=$3.}
\label{tab.post.inf}
\end{table}

\newcommand{\otherseltables}{
}

We see that, for the selected second coefficient, the confidence intervals in Table \ref{tab.post.inf} show a significant difference from behavior . 
In the same way that we showed our theoretical estimators for the coefficients and value in Figure \ref{fig:sim} against an empirical reference, we do the same for the confidence intervals in Table \ref{tab.post.inf}. We know that the coefficients that are set to their behavioral counterparts have nominal coverage, because they are derived by maximum likelihood estimation \citep{casella2002statistical,van2000asymptotic}, but we must check the coverage for the non-behavioral coefficient corresponding to $S_{t,2}$.
In Section \ref{sec:simCoverage}, we do so by conducting an additional Monte-Carlo study in which we generate Table \ref{tab.post.inf} $M_C=500$  times, and check whether the confidence interval coverage for $\beta_{0,\gamma,2}$ is nominal. 

\subsection{Coverage of active parameters}
\label{sec:simCoverage}

We assess the results in Table \ref{tab.post.inf} with a Monte-Carlo study. In the process, we assess our estimator, $\sigma_n^2$, for the variance of $\beta_{n,\gamma}$ (Theorem \ref{thm:betanConsistNorm} (ii) and Equation (\ref{eq:sigma2n})).  
Details of the Monte-Carlo summary statistics are given in Appendix \ref{app:simsum}, but we summarize them briefly here.  To assess coverage, we must  know $\beta_{0,\gamma}.$ However, in this policy search setting,  $\beta_{0,\gamma}$ is unknown to us, even when we are simulating the data.  Since we specified the functional form of the reward to be (\ref{eq:rewardspec}), however, we can estimate $\beta_{0,\gamma}$ arbitrarily well by sidestepping the need for importance sampling, as described in Appendix \ref{supp.section.estimating.beta.gamma}. We thus obtain the ``true'' parameter, $\beta_{0,n,\gamma}$, defined in (\ref{eq:mc:beta0ngamma}) of Appendix \ref{supp.section.estimating.beta.gamma}. We index $\beta_{0,n,\gamma}$ with $0$ to designate its role as a reference against which we will check our theory, and we treat  $\beta_{0,n,\gamma}$ as $\beta_{0,\gamma}$ when we evaluate coverage. 

To check consistency (Theorem \ref{thm:betanConsistNorm} (i)), we also compare the ``true'' estimand $\beta_{0,n,\gamma}$ to the estimated coefficients, $\beta_{n,\gamma}=\arg\max {M}_n,$ averaged over Monte-Carlo datasets, which we denote $\bar{\beta}_{n,\gamma},$ and define in (\ref{eq:mc:betanbar}) in Appendix \ref{supp.section.estimating.beta.gamma}.
  We also derive a Monte-Carlo estimator for the standard deviation of the estimator $\beta_{n,\gamma},$ which we call  $\sigma_n(\beta_{0,n,\gamma}),$ and which is just  the standard deviation of the coefficient estimates over Monte-Carlo datasets. More detail on the computation of this quantity is given in Appendix \ref{app:simsum}, Equation (\ref{eq:mc:sigma0n}). We include a subscript 0 in $\sigma_{0,n}(\beta_{0,n,\gamma})$ to designate its role as a ``true reference'' against which we will check our theoretical results. 
 
 Define the estimated standard deviation that we obtain from (\ref{eq:sigma2n}) based on Theorem \ref{thm:betanConsistNorm} (ii), as $\bar{\sigma}_n(\beta_{n,\gamma}),$ where the overbar indicates that this estimated variance was computed for each Monte-Carlo dataset and then averaged (more detail is given in Appendix \ref{app:simsum}, Equation (\ref{eq:mc:barsigman})).  We will compare the estimated standard error, $\bar{\sigma}_n(\beta_{n,\gamma}),$ to  the ``true'' standard error, $\sigma_{0,n}(\beta_{0,n,\gamma}),$ and we will also assess coverage.
 Note that we are simulating ``post-selection'' inference, so we only assess coverage for selected coefficients, where the selection was made in Figure \ref{fig:sim}, and, in this case, concerns only $\beta_{n,\gamma,2}$, which corresponds to the second covariate, $S_{t,2}$.

\newcommand{\simcoef}{
\begin{proof}
Note that 
\begin{align*}
V_0 &= E\left(\sum_{t=0}^T \gamma^t R(S_t,A_t,S_{t+1})\right)\\
&=E\left(\sum_{t=0}^T \gamma^t E(R(S_t,A_t,S_{t+1})|S_t)\right)\\
&=-E\left(\sum_{t=0}^T \gamma^t E(S_{t,2} A_t|S_t)\right) \ \ \ \ \ \ \ (\text{definition of } R)\\
&=-E\left(\sum_{t=0}^T \gamma^t (  S_{t,2} E(A_t|S_t))\right) \ \ \ \ \ \ \ (\text{definition of $S_{t+1}$ })\\
&=-E\left(\sum_{t=0}^T \gamma^t ( S_{t,2} \expit(\beta^TS_t))\right)\\
\end{align*}
This is maximized by $\pi_{\beta_0}(A_t=1|S_t)=I(S_{t,2}<0),$  which occurs if $\beta_{0,2}=-\infty.$   In simulation, we will see a large negative coefficient $\beta_2$, which will shrink in magnitude as we increase $\gamma$ and $\lambda$.
\end{proof}
}
 
\newcommand{\simmetrics}{
Recall that
$\beta_{n,\gamma}=\arg\max_{\beta} {M}_n.$ For Monte-Carlo dataset $m$ let the estimator obtained by maximizing ${M}_n$ be  $\beta_{n,\gamma}^{(m)}.$
Over $M$ Monte-Carlo datasets, define the average to be 
\begin{equation}
\label{eq:mc:betanbar}
\bar{\beta}_{n,\gamma}=\frac{1}{M}\sum_{m=1}^M \beta_{n,\gamma}^{(m)}.
\end{equation} Also recall that $\sigma^2_n$ is our  asymptotic variance estimator, defined in (\ref{eq:sigma2n}), let this estimator derived from Monte-Carlo dataset $m$ be $(\sigma^2_n)^{(m)},$ and define the average over Monte-Carlo datasets as 
\begin{equation}
\label{eq:mc:barsigman}
\bar{\sigma}_n^2(\beta_{n,\gamma})=\frac{1}{M}\sum_{m=1}^{M}(\sigma^2_n)^{(m)}.
\end{equation}
 Further, define the empirical variance of the estimated $\beta_{n,\gamma}$ as 
\begin{equation}
\label{eq:mc:sigma0n}
\sigma^2_{0,n}(\beta_{n,\gamma})=\frac{1}{M-1}\sum_{m=1}^M( \sqrt{n}\beta_{n,\gamma}^{(m)}-\bar{\beta}_{n,\gamma,\sigma})^2,
\end{equation}
where $\beta_{n,\gamma}^{(m)}$ is the solution for Monte-Carlo iteration $m$ and 
\[\bar{\beta}_{n,\gamma,\sigma}=\frac{1}{M}\sum_{m=1}^M \sqrt{n}\beta_{n,\gamma}^{(m)}.\] Note that we multiply $\beta_{n,\gamma}$ by $\sqrt{n}$ so that we can estimate $\sigma$, the standard deviation, not  $\sigma/\sqrt{n},$ the standard error.  When we construct confidence intervals, we then have to divide our estimate of $\sigma$ by $\sqrt{n}.$ We index (\ref{eq:mc:sigma0n}) by $0$ to designate its role as a ``reference'' quantity, against which we will check our theory.
}


\begin{table}[ht]
\centering
\begin{tabular}{llll}
$\gamma$ & 0.01 & 3.00 & 6.00 \\ 
  \hline
  True: $\bar{\beta}_{0,n,\gamma}$ & -6.33 & -0.13 & 0.03 \\ 
  Estimated: $\bar{\beta}_{n,\gamma}$ & -6.35 & -0.13 & 0.03 \\ 
  Bias & -0.03 & 0.00 & 0.00 \\ 
  True: $\sigma_{0,n}(\beta_{n,\gamma})$ & 13.62 & 1.60 & 1.50 \\ 
    Estimated: $\bar{\sigma}_n(\beta_{n,\gamma})$ & 16.08 & 1.53 & 1.41 \\ 
  Coverage & 0.98 & 0.95 & 0.93 \\ 
  Length CI & 2.82 & 0.27 & 0.25 \\ 
\end{tabular}
\caption[Coverage]{ {\bf Bias, standard deviation, and coverage for the coefficient, $\beta_{0,\gamma,2}$, of the selected covariate, $S_{t,2}$.} For simulation settings $n = $ 500, $T = $ 2, 
            $K = $ 2, and $M_C = $ 500, we show these performance measures while varying  $\gamma$. We also show the estimated (indexed by $n$ alone) and ``true'' (indexed by $0$ and $n$, indicating an empirical estimate of the true value) coefficients, $\bar{\beta}_{n,\gamma}$ and $\bar{\beta}_{0,n,\gamma}$, and standard deviations, $\bar{\sigma}(\beta_{n,\gamma})$ and $\sigma_{0,n}(\beta_{n,\gamma})$, where the overbar indicates an average over Monte-Carlo datasets. 
            } 
\label{tab:cov}
\end{table}

\newcommand{\othercovtables}{

}
 Results are shown in Table \ref{tab:cov}, where we see roughly nominal coverage ($0.95$) for the active covariate in this problem, supporting Theorem \ref{thm:betanConsistNorm} (ii) and the corresponding theoretical variance from (\ref{eq:sigma2n}). 
 We see over-coverage for small $\gamma.$  As discussed in Section \ref{sec:selectionDiagramsSim}, the parameter $\gamma$  impacts the stability of the objective.   If $\gamma$ is too small, ${V}_n$, which is unstable, will dominate the objective in (\ref{eq:mn}), and we will gain value but lose stability. If $\gamma$ is larger, $KL_n,$ which is more stable (as discussed in Section \ref{sec:behcon}), will be more prominent, and we  will lose value, but we will gain stability.  However, note that 
 while the coverage for the smallest $\gamma$ is conservative (the variance is over-estimated), it is not as severe as one would expect when viewing the selection diagram for the corresponding $\gamma$ in the top row and middle column of Figure \ref{fig:sim}, in which it appears that the coefficient estimates are quite unstable.  This is because, as discussed in Sections \ref{sec:estTuning} and \ref{sec:estTuning}, there are fewer degrees of freedom and a larger sample after selection, because some coefficients are fixed to their behavioral counterparts, and an entire half of the dataset is devoted to post-selection inference  (rather than one quarter, which is the fraction devoted to selection).  

\section{Real data analysis}
\label{realdata}
We illustrate the proposed methodology and theory on a real dataset generated by patients and their healthcare providers in the intensive care unit, as in {\citerelativesparsity}.  We show that we can derive a relatively sparse policy and perform inference for the coefficients.
\subsection{Decision problem}
\label{real.data.problem.statement.inference}
We consider the same real data decision problem as in {\citerelativesparsity}, but in the multi-stage setting.
There is variability in vasopressor administration in the setting of hypotension \citep{der2020narrative,russell2021vasopressor,lee2012interrogating}.  Although vasopressors can  stabilize blood pressure,  they have a variety of adverse effects, making vasopressor administration an important decision problem. We extend code from {\citerelativesparsity}, which was derived from \cite{futoma2020popcorn,gottesman2020interpretable}, to process the freely available, observational electronic health record dataset, MIMIC III  \citep{johnson2016mimic2, johnson2016mimic1,goldberger2000physiobank}.  We include patients from the medical intensive care unit (MICU), as in {\citerelativesparsity}.
  We illustrate how we can provide inference for a relatively sparse policy in this setting.

As in {\citerelativesparsity}, we begin the trajectory at the beginning of hypotension, which is defined as a mean arterial pressure (MAP) measurement that is less than $60,$ a cutoff used in  \cite{futoma2020popcorn}. We consider the first 45 minutes after hypotension onset, where the first 15 minutes is $S_0$, the second 15 minutes $S_1,$ and the third $S_2.$  Eleven patients left the MICU before 45 minutes, so we excluded those patients.  Actions will be $A_0,$ taken after observing $S_0$ (we take the last measured covariates in that time window), and $A_1,$ taken after observing $S_1.$  We restrict ourselves to two stages, because it is important to stabilize MAP early, and, also, because patients often leave the ICU due to death or discharge (if we instead used e.g., 10 stages, many patients would leave, leading to more missingness). We consider any vasopressor administration to be an action, and different vasopressors are aggregated and normalized as in \cite{futoma2020popcorn,komorowski2018artificial} to their norepinephrine equivalents. As in {\citerelativesparsity}, because the norepinephrine duration of action is short, we assume that a vasopressor administered in each 15 minute interval does not affect the blood pressure at the end of the following 15 minute interval.

We consider the same set of covariates as those in {\citerelativesparsity} and \cite{futoma2020popcorn}, which includes MAP, heart rate (HR), urine output (Urine), lactate, Glasgow coma scale (GCS), serum creatinine, fraction of inspired oxygen (FiO2), total bilirubin, and platelet count.  As in {\citerelativesparsity}, based on \cite{futoma2020popcorn},  extreme, non-physiologically reasonable values of covariates were floored or capped, and missing data was imputed using the median or last observation carry forward.  In terms of the reward, we have that $S_t$ contains MAP as its first component; hence, we define 
$R(S_t,A_t,S_{t+1})=S_{t+1,1}.$
 In other words, we define the reward as 
$R(S_t,A_t,S_{t+1})=R((MAP_t,\dots)^T,A_t,(MAP_{t+1},\dots)^T)=MAP_{t+1},$
where the notation $(MAP_t,\dots)$ indicates that the state depends on other covariates besides $MAP$. 
This reward reflects the short-term goal of increasing blood pressure in the setting of hypotension, and vasopressors should increase this reward.   This reward is imperfect but sensible, and, as discussed in {\citerelativesparsity}, by constraining to the behavioral policy (the standard of care), we are able to derive a suggested policy that improves outcomes with respect to this sensible reward, but does not exclusively maximize this imperfect reward.

After excluding 39 patients who left the intensive care unit within the first 45 minutes of receiving a $MAP < 60$ (which promps entry into the cohort), we had $n=11,715$ patients.
As in {\citerelativesparsity}, we start by checking that the behavioral policy model in (\ref{eq:thetapolicyInf}) is specified correctly.  To this end, we estimate a calibration curve \citep{van2016calibration,niculescu2005predicting}, which is shown in
Figure \ref{calcurveInference} of Appendix \ref{app:calcurve} and suggests that the model specification in (\ref{eq:thetapolicyInf}) is reasonable.

\subsection{Selection diagrams}
In Figure \ref{fig:real.data}, we show selection diagrams, as we did in the simulations.
In Figure \ref{fig:real.data}, we see that, for fixed $\delta$ and $\gamma,$  the suggested policy coefficients, $\beta_{n,\gamma,\lambda},$ approach their behavioral counterparts as $\lambda$ increases.  We see that, like in {\citerelativesparsity}, MAP is isolated, and a relatively sparse policy is derived that still has value that is one standard error above the behavioral value. The selected $\lambda$ of this relatively sparse policy, which we denote $\lambda_n,$ is shown as the dotted, vertical line in the central triplet of Figure \ref{fig:real.data}, and was determined using (\ref{lambda.crit.diff}).  

The shaded regions in the coefficient panels show the variability of $\beta_{n,\gamma,\lambda}$; we see that there is considerable variability, especially with small $\gamma$. As noted in section \ref{sec:selectionDiagramsSim}, this variability is likely exaggerated, since the coefficients are estimated on only one half of the data, which is further split in half.  In the post-selection step, we will have a larger sample and fewer covariates, both of which stabilize the problem, as discussed in Sections \ref{sec:estTuning} and \ref{sec:simCoverage}.
 We also see that the standard error of $V_n$ (shaded) is small for large $\lambda,$ where the suggested and the behavioral policy are the same. It is the standard error in this region that is  used to select $\lambda$ in (\ref{lambda.crit.diff}).  
 Given a selection, $\lambda_n,$ we now perform post-selection inference in a held-out split of the data, results for which are shown in Table \ref{tab:real.data.post.select.coef}. 
\begin{figure}[htp!]
\centering
\includegraphics[width=0.97\textwidth]{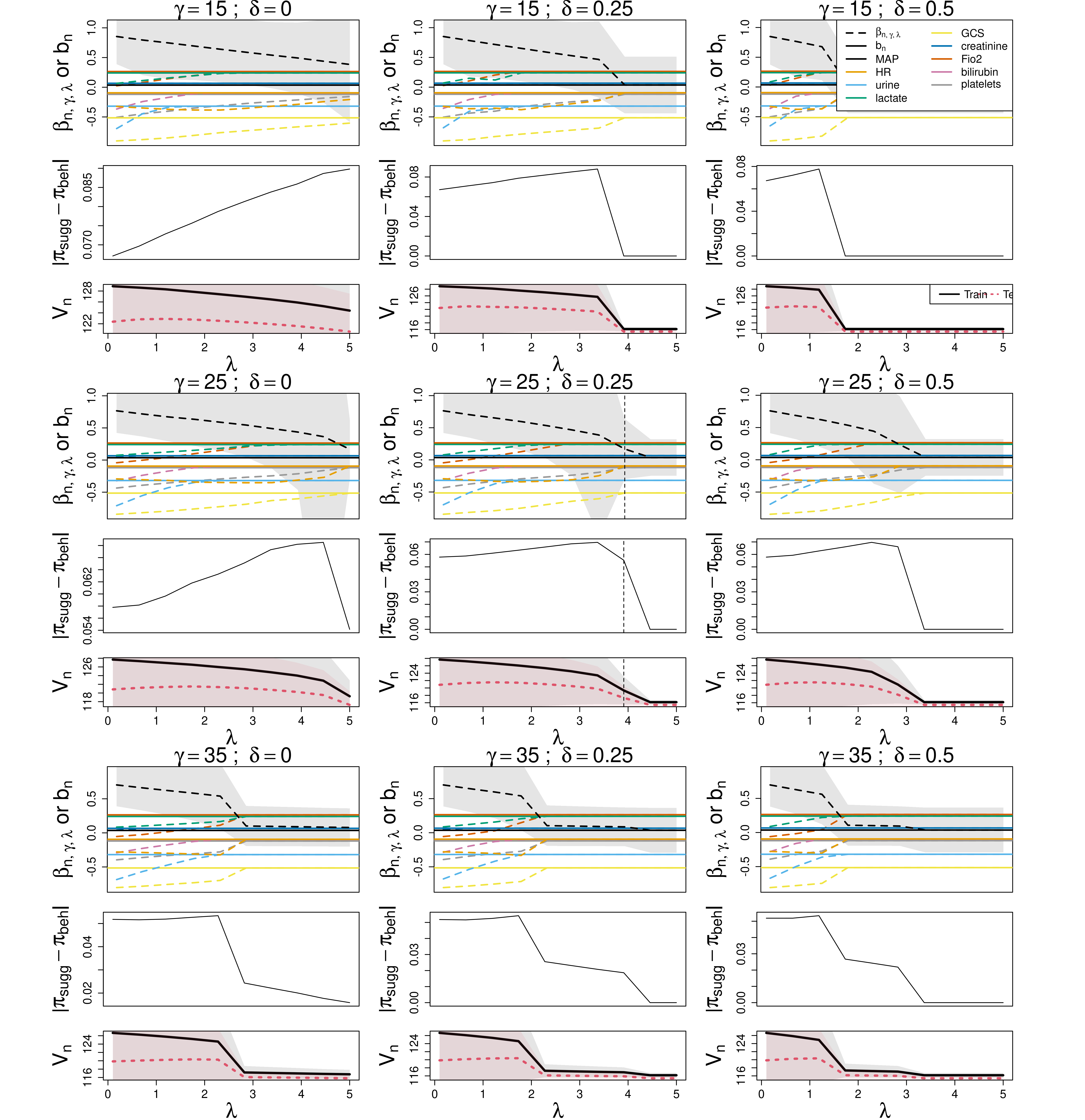}
\caption[Selection diagrams for the real data]{{\bf Selection diagrams for the real data ($n_{train}=1,176,n_{test}=1,176$).}   Over increasing $\gamma$ (going down) and $\delta$ (going right), we show the average coefficients in the suggested ($\beta_{n,\gamma,\lambda})$ and behavioral ($b_n$) policies, the average difference in probability of treatment between the two policies, and the average value ($V_n$) for the suggested policy, all of which were computed in the first split of the data . The dotted vertical line indicates, $\lambda_n$, a choice of $\lambda$  based on (\ref{lambda.crit.diff}). Note that the average suggested policy probability of treatment is $\pi_{sugg}=(1/nT)\sum_i\sum_t{\pi}_{\beta_{n,\gamma,\lambda}}(A_{i,t}=1|s_{i,t})$ and vice versa for $\pi_{beh}$.
The shaded regions in the coefficient ($\beta_{n,\gamma,\lambda}$) panels correspond to (\ref{eq:sigma2n}) (to declutter the plot, and because MAP was the only selected covariate, we show this only for MAP).
The shaded regions in the value panels show one standard error based on (\ref{valuevar}), which was used to select $\lambda_n$ based on (\ref{lambda.crit.diff}).
 } 
\label{fig:real.data}
\end{figure}
\subsection{Post-selection inference}
  As shown in Figure \ref{fig:real.data}, we selected $\gamma,$ $\lambda,$ and $\delta$ in the first split of the real data.  Given this selection, we now  perform post-selection inference on the second split of the data. We report results in Table \ref{tab:real.data.post.select.coef}. 
  In particular, as in {\citerelativesparsity}, we see that all coefficients except one are fixed to their behavioral counterparts,  making it easy to discuss and justify the suggested policy to the patients and providers who may choose to adopt it. Unlike in {\citerelativesparsity}, we now have a 95\% confidence interval for the coefficient for MAP, which was derived from Theorem \ref{thm:betanConsistNorm} (ii) and the corresponding estimator (\ref{eq:sigma2n}). Note that the confidence interval shown in Table \ref{tab:real.data.post.select.coef} is much narrower than the shaded region shown in Figure \ref{fig:real.data}; the former was rigorously derived based on Theorem \ref{thm:betanConsistNorm} (ii), whereas the latter was just a heuristic way to visualize variability. Also note that, because we are in the post-selection setting, the sample size used to derive the confidence interval in Table \ref{tab:real.data.post.select.coef} is twice as large as that used to estimate (\ref{eq:sigma2n}) in Figure \ref{fig:real.data}, and, because all but one parameter in the post-selection step are set to their behavioral counterparts, there is one free parameter in Table \ref{tab:real.data.post.select.coef}, while there are nine free parameters in Figure \ref{fig:real.data}.      

\begin{table}[htp!]
\centering
\begin{tabular}{lll}
 & Suggested ($\beta_{n,\gamma})$ & Behavioral ($b_n$) \\ 
  \hline
MAP & 0.132 (-0.136, 0.400) & -0.070 (-0.159, \ 0.019) \\ 
  HR & set to $b_n$ & \ 0.052 (-0.069, \ 0.173) \\ 
  urine & set to $b_n$ & -0.300 (-0.478, -0.122) \\ 
  lactate & set to $b_n$ & \ 0.426 (\ 0.307, \ 0.546) \\ 
  GCS & set to $b_n$ & -0.585 (-0.679, -0.490) \\ 
  creatinine & set to $b_n$ & \ 0.186 (\ 0.080, \  0.293) \\ 
  Fio2 & set to $b_n$ & \ 0.118 (-0.004, \ 0.241) \\ 
  bilirubin & set to $b_n$ & \ 0.033 (-0.065, \ 0.131) \\ 
  platelets & set to $b_n$ & -0.055 (-0.176, \ 0.065) \\ 
\end{tabular}
\caption[Real data post-selection inference]{Estimated coefficients (95\% confidence interval) for held out dataset post-selection inference; $n=$2,352  and $\gamma=$25.} 
\label{tab:real.data.post.select.coef}
\end{table}

 In Table \ref{tab:real.data.post.select.coef}, we see that the confidence interval for the suggested policy coefficient for MAP is shifted more toward a positive range, and is much wider, with considerable negative and positive margins, than the confidence interval for the behavioral policy coefficient for MAP, which is narrow and almost entirely within a negative range. 

\section{Discussion}
In {\citeinference}, we developed methodology and theory to enable inference when using  the relative sparsity penalty developed in {\citerelativesparsity}. This inference framework allows one to construct confidence intervals for the relative sparsity coefficients, improving the rigor of the method, and ultimately facilitating safe translation into the clinic. To our knowledge, we are the first to fully characterize the difficulties, such as infinite estimands, with performing inference in the policy search setting under a generalized linear model policy.  We created finite, behavior-constrained estimands by repurposing a weighted version of Trust Region Policy Optimization (TRPO) \citep{schulman2015trust} as the ``base'' objective within the relative sparsity framework.  We proved novel theorems for weighted TRPO and operationalized these results toward inference in a sample splitting framework, which is free of issues associated with post-selection inference. Unlike standard sample splitting techniques, our framework required that we set non-selected parameters to some value other than zero, which considerably complicated the partial derivatives of the objective functions, since the nuisance began to appear in non-standard locations such as the numerator of the inverse probability weighting ratio. We then developed an adaptive relative sparsity penalty, which improved the discernment of the penalty.  
We developed all of our methodology and theoretical results for the observational data setting in the multi-stage, generalized linear model framework. Finally, we illustrated our framework for inference for the relative sparsity penalty on an intensive care unit, electronic health record, for which estimation was non-trivial, and inference was even more difficult. In simulations, and in the real data, we rigorously characterized sensitivity of the proposed inference framework to the tuning parameters. We finally presented selection diagrams, which are tools that help to select the tuning parameters using training data. 


There are several opportunities for future work.  For example, although we have developed our method for large observational datasets, the current sample splitting scheme could still be improved to make better use of the data (e.g., a bootstrap could be performed, rather than a single sample split).  Also, as discussed in {\citerelativesparsity},  the real data analysis could be refined by including a reward that takes into account mortality and morbidity.  Further, as discussed in {\citerelativesparsity}, relaxing the linearity assumption of the behavioral policy model (although we also checked the reasonableness of this specification here) would be a good direction for future work. The assumption that vasopressors administered in one time step have a negligible effect on the MAP observed at the end of the next time step is perhaps overly simplified, even though intravenous vasopressors have a short duration of action.  In general,  more proximal administrations might have more of an impact than more distal administrations. Discretizing time is also a considerable simplification.
 
Other assumptions that we make in {\citeinference}, such as the global uniqueness of the maximizer of the objective function, might be overly restrictive; it would be useful to consider a local uniqueness instead. Also, it would be interesting to explore other ``base'' objective functions, such the weighted likelihood objective of \cite{ueno2012weighted}, which may have favorable properties in this regard. 
 The stationarity of the behavioral policy could  be relaxed by indexing the policy parameters by time step. The Markov property of the behavioral policy might be more challenging to relax.  Unlike many other methods, we do not make Markov and stationarity assumptions for the transition probabilities. 
 One could  increase the likelihood that the  no unmeasured confounders assumption holds by adjusting for more covariates, which are often available in large electronic health record datasets.  
 There are always challenges associated with observational data, including unverifiable assumptions.  We emphasize that any suggested treatment from a policy derived with the proposed method should be reviewed by the medical care team. As discussed in \cite{relsparSIM}, the transparency of the relative sparsity framework facilitates this type of review, and the methodology and theory for inference provided in our work increases the rigor of the relative sparsity framework. 

 \section{Acknowledgements}
The authors thank Jeremiah Jones, Ben Baer,  Michael McDermott, Brent Johnson, and Kah Poh Loh for helpful discussions.
This research, which is the sole responsibility of the authors and not the National Institutes of Health (NIH), was supported by the National Institute of Environmental Health Sciences (NIEHS) and the National Institute of General Medical Sciences (NIGMS) under T32ES007271 and T32GM007356.

\section{Conflict of interest}
Authors state no conflict of interest.

\bibliographystyle{plainnat}
\bibliography{referencesnew}


\appendix
 \renewcommand\thefigure{\thesection.\arabic{figure}}  
 \setcounter{figure}{0}   


\section{Appendix}

\subsection{Research code}
Research code can be found at \url{https://github.com/samuelweisenthal/inference_for_relative_sparsity}. 

\subsection{Data availability statement}
The MIMIC \citep{johnson2016mimic} dataset that supports the findings of this study is openly available at PhysioNet (doi: \url{10.13026/C2XW26}) and can be found online at \url{https://physionet.org/content/mimiciii/1.4/}.
\subsection{Determinism of the reward-maximizing policy}
\label{supp:proverDeterm}
\begin{lemma}
\label{lemma:det}
Assuming (\ref{eq:thetapolicyInf}), which defines the policy as an $\expit,$ and defining the true value maximizing policy $\pi_{\beta_0} = \arg\max_{\pi} V_0,$ let $R_m = R(S_m,A_m,S_{m+1})$  we have that \begin{align*}
\pi_{\beta_0}(A_t=1|s_t)
&= I\left(E\left(\sum_{m=t}^T R_m|S_t,A_t=1\right) -E\left(\sum_{m=t}^T R_m|S_t,A_t=0\right)>0\right).
\end{align*}
\end{lemma}
\begin{proof}
We will show that the policy that informs the decision $\pi(A_t=1|S_t)$ is equal to one or zero (i.e., it is deterministic). This has been shown in e.g. \cite{lei2017actor,puterman2014markov}, but we specialize the proof to the binary action, continuous state setting. Consider for now an arbitrary policy $\pi$, which is not parameterized.  We will show that the arbitrary policy must be deterministic, and then we will argue that in order for the model that we use for the policy in this work, defined in (\ref{eq:thetapolicyInf}), to approximate this deterministic policy, at least some of the policy coefficients must approach infinity in magnitude. 

We start by showing that, to maximize value, an arbitrary policy $\pi$ must be deterministic.  Let us consider the value starting at time $t$ and going forward,
\begin{align*}
  V_{0,t:T} &= E\left(\sum_{m=t}^T R(S_m,A_m,S_{m+1})\right)\\
  &= E\left(E\left(\sum_{m=t}^T R(S_m,A_m,S_{m+1})|S_t,A_t\right)\right)\\
 &=\int_{s_t}\sum_{a_t}E\left(\sum_{m=t}^T R(S_m,A_m,S_{m+1})|s_t,A_t\right)\pi(a_t|s_t)p(s_t)ds_t\\
  &=\int_{s_t}E\left(\sum_{m=t}^T R_m|s_t,A_t=1\right)\pi(A_t=1|s_t)p(s_t)ds_t\\&+ E\left(\sum_{m=t}^T R_m|s_t,A_t=0\right)\pi(A_t=0|s_t)p(s_t)ds_t.
      \end{align*}
    Where we have just written out the expectation. Now, the last expression is equal to
    \begin{align*}
    &\int_{s_t}E\left(\sum_{m=t}^T R_m|s_t,A_t=1\right)\pi(A_t=1|s_t)p(s_t)ds_t\\&+ E\left(\sum_{m=t}^T R_m|s_t,A_t=0\right)(1-\pi(A_t=1|s_t))p(s_t)ds_t\\
        &=\int_{s_t}E\left(\sum_{m=t}^T R_m|s_t,A_t=1\right)\pi(A_t=1|s_t)p(s_t)ds_t\\&
        + E\left(\sum_{m=t}^T R_m|s_t,A_t=0\right)p(s_t)ds_t\\
        &-E\left(\sum_{m=t}^T R_m|s_t,A_t=0\right)\pi(A_t=1|s_t)p(s_t)ds_t.
        \end{align*}
   Where we have used the properties of model (\ref{eq:thetapolicyInf}). This is further equal to 
    \begin{align*}
             &\int_{s_t}\left(E\left(\sum_{m=t}^T R_m|s_t,A_t=1\right)  - E\left(\sum_{m=t}^T R_m|s_t,A_t=0\right)\right)\pi(A_t=1|s_t)p(s_t)ds_t\\&
             + \int_{s_t}E\left(\sum_{m=t}^T R_m|s_t,A_t=0\right)p(s_t)ds_t\\
&=\int_{s_t}\left(E\left(\sum_{m=t}^T R_m|S_t,A_t=1\right) -E\left(\sum_{m=t}^T R_m|S_t,A_t=0\right)\right)\pi(A_t=1|s_t)p(s_t)ds_t\\&\ \ +C
\end{align*}
where $C$ does not depend on $\pi(A_t=1|s_t).$
This implies that, since $0\leq \pi_{\beta_0}(A_t=1|s_t)\leq 1$, the maximizing policy is
\begin{align}
\label{indi}
&\pi_{\beta_0}(A_t=1|s_t)= I\left(E\left(\sum_{m=t}^T R_m|S_t,A_t=1\right) -E\left(\sum_{m=t}^T R_m|S_t,A_t=0\right)>0\right).
\end{align}
Note that the expected return at time $t$ depends on the  actions taken  from time $t$ to $T, $ and so if the policy is stationary (Assumption \ref{assum:stationarity}), this action will depend on the policy itself.  We therefore start at time $T$ and go backwards, applying (\ref{indi}) at each time step, so that the policy at time $t$ is an indicator, and depends on a series of indicators that determine actions taken up until time $T$. Hence, at any time step, the optimal policy is deterministic.

We will now show that this determinism has implications for the magnitude of the policy coefficients (to approximate an indicator function, an $\expit$ must have arbitrarily large input, which can only occur if the coefficients of the policy are arbitrarily large).
Note that \[\pi_{\beta_0}(A_t=1|s_t) = \expit(\beta_0^Ts_t)\] and 
\[\expit(\beta_0^Ts_t)=1 \iff \exp (\beta_0^Ts_t)=1+\exp (\beta_0^Ts_t),\] which is satisfied only if $\exp (\beta_0^Ts_t)=\infty,$ which is satisfied if and only if, because $s_t$ is bounded, some entry of $\beta_0$ is infinite (or, in practice, becomes arbitrarily large).  
\end{proof}

\subsection{The expected value of the importance sampling ratio}
\label{app:proveEr}
\proveEr

\subsection{On the adaptive lasso tuning parameter}
\label{app:deltameaning}
\deltameaning

\subsection{Scaling of the state covariates}
\label{app:scaling}
We scale independently of time, because the policy coefficients are constant over time by stationarity (Assumption \ref{assum:stationarity}).  Concretely, for  $i=1,\dots,n,\ t=1,\dots,T,\ k=1\dots,K,$ if $\sigma_{n}(S_k)$ is an estimate for the standard deviation of the dimension $k$ of the state, we scale $S_{itk}$ as
$\tilde{S}_{itk}={S_{itk}}/{\sigma_{n}(S_k)}.$ 
 Since we will be penalizing the coefficients of the suggested policy to the behavioral policy coefficients, and the $K$ covariates that comprise the state $S_t=(S_{t,1},\dots,S_{t,K})^T$ must be scaled to put the coefficients on a level playing field for penalization, we also scale before estimating the behavioral policy $b_n$. 
Scaling is also discussed in \cite{relsparSIM,hoffman2011regularized}. 

\subsection{Additional notation}
We first briefly introduce some notation that was not necessary in the methodology section, but allows for more compact proofs here. Recall that $E$ and $E_n$ are the expectation and the empirical expectation operator, respectively. 
Define the log of the data distribution, which can be written by Assumptions \ref{assum:stationarity} and \ref{assum:markov} as
\begin{equation}
\label{eq:l}
l(b)=\log P_0(S_0)\pi_{b}(A_0|S_0)\prod_{t=1}^{T}\pi_{b}(A_t|S_t)P_{0,t}(S_t|A_{t-1},S_{t-1},\dots,A_0,S_0).
\end{equation}
Note that the behavioral policy parameter estimator is
\begin{equation}
b_n=\arg\max_{b}nE_n l(b).
\end{equation}
Define
\begin{equation}
\label{eq:r}
r = {\prod_{t=0}^T \pi_{\beta}(A_t|S_t)}\biggr/{\prod_{t=0}^T \pi_{b}(A_t|S_t)}. 
\end{equation}
Note that we use lowercase $r$ to refer to the ratio here, which should not be confused with capital $R,$ which refers to the reward in Equation (\ref{eq:R}).
Note that for $V_0$ defined in (\ref{eq:V0}) and rewritten in (\ref{eq:V0times1}), $V_0(\beta)=Ev,$ 
where 
\begin{align}
\label{eq:v}
v=r \sum_{t=0}^{T}R(S_t,A_t,S_{t+1}).
\end{align}
Although $V_0$ as written in (\ref{eq:V0times1}) depends on $b,$ we suppress this dependence on $b,$ because the policies involving $b$ cancel.
Then ${V}_n(\beta,b)={E_n v(\beta,b)}/{E_n r (\beta,b)}.$ 
Define the standard, unweighted importance sampling estimator to be
\begin{align}
\label{eq:vnestInference}
\tilde{V}_n(\beta,b) 
=E_n v(\beta,b).
\end{align}
Note that for $KL_0$ defined in (\ref{eq:KL0}),
$KL_0({\beta},b)= E kl,$
where 
\begin{equation}
\label{eq:kl}
kl=\log\left(1/r\right).
\end{equation}
Note that for $KL_n$ defined in (\ref{eq:KLn}),
$KL_n(\beta,b)=E_n kl$.
Note that if we define 
\begin{equation}
\label{eq:m}
m=v-\gamma kl, 
\end{equation}
then
$M_0(\beta,b)=Em.$
Note that for ${M}_n$ defined in (\ref{eq:mn}),
 ${M}_n(\beta,b_n,\gamma)= {V}_n(\beta,b_n)-\gamma KL_n(\beta,b_n)
 ={E_nv}/{E_n r} - \gamma E_n kl
= E_n \left ({v}/{E_n r} -\gamma  kl\right).$

\subsection{Estimating the behavioral policy}
\label{app:estimatingBehavioral}
Recall that $E_n$ is the empirical expectation operator, where $E_n f(X) = \frac{1}{n}\sum_{i=1}^nf(X_i)$.
Recall the sampling distribution defined in (\ref{eq:l}), which was \[l(b)=\log P_0(S_0)\pi_{b}(A_0|S_0)\prod_{t=1}^{T}\pi_{b}(A_t|S_t)P_{0,t}(S_t|A_{t-1},S_{t-1},\dots,A_0,S_0).\] Define the log likelihood then as 
\begin{equation}
nE_n l = \sum_{i=1}^n\log P_{0,b}(a_{i,0:T},s_{i,0:T+1}).
\end{equation}
If we consider the maximizer,
\begin{equation}
b_n=\arg\max_{b}nE_n l(b),
\end{equation}
then $b_n$ is an estimator for $b_0,$ the true behavioral policy parameter,
and $\pi_{b_n}$ is an estimator for the data-generating ``behavioral'' policy,  $\pi_{b_0}$, where the form of $\pi$ is an $\expit,$ as given in (\ref{eq:thetapolicyInf}).

\subsection{Estimating value with importance sampling}
\label{est.v0.Precup}
\IS


\subsection{Proof of Lemma \ref{lemma:MnConsist}}
\label{app:prf:lemma:MnConsist}
We can by Assumption \ref{as:bconsist} on consistency of the nuisance, perform a Taylor expansion around $b_0,$ so that  
\begin{align*}
{{M}_n}(\beta,b_n)&\approx {{M}_n}(\beta,b_0) + (b_n-b_0)^T \fdr{b}{{M}_n}(\beta_{n,\gamma},b_0)\\
&={{M}_n}(\beta,b_0) + o_P(1) \fdr{b}{{M}_n}(\beta_{n,\gamma},b_0)\\
&={{M}_n}(\beta,b_0) + o_P(1).
\end{align*} 
where the last equality follows due to boundedness of the derivative by Assumption \ref{as:mbounded} and the Central Limit Theorem, since the sample derivative is an average.
Hence, we only consider the convergence of ${M}_n(\beta,b_0).$ 
Note that ${M}_n = {V}_n - \gamma KL_n.$ We have that ${V}_n$ converges to $V_0$ by Remark \ref{rem:Vnconsist},  and $KL_n\gop KL_0$ by the Law of Large Numbers.

\subsection{Proof of Lemma \ref{lem:limDerivsRn}}
\label{app:limDerivs}
Let us show that $\fdr{\beta}{E_n r}\gop 0.$  
 By Assumption \ref{as:mbounded}, if we just set $\gamma=0$ and $R=1$, we have that $|\fdr{\beta}{r}|< C,$ for some constant $C$, which allows us to apply Lebesgue's Dominated Convergence Theorem. Then
\begin{align*}
E_n \fdr{\beta}{r} &\gop E \fdr{\beta}{r}, \text{\ \ \ \ \ \ \ \ \ \      (Law of Large Numbers)}\\
&=\fdr{\beta}{E}{r} \text{ \ \ \ \ \ \ \  \ \ \ \             (Dominated Convergence Theorem) }\\
&= \fdr{\beta}{1} \text{\ \ \ \ \ \ \ \ \ \ \ \ \ \ \ (Remark \ref{rem:Er1})}\\
&=0.
\end{align*}
Identical arguments can be made to show that $E_n \fdr{b}{r}\gop 0$ and $E_n \sddr{b}{\beta}{r}\gop 0.$



\subsection{Proof of Lemma \ref{lem:limDerivsMn}}
\label{app:lem:limDerivsMn:prf}
We first show that $J_n=\fdr{\beta}{{M}_n}\gop \fdr{\beta}{M_0}=J_0.$ 

Note that if we write the arguments explicitly, we have $\fdr{\beta}{{M}_n}(\beta,b_n).$  We can, however, by Assumption \ref{as:bconsist} on consistency of the nuisance, perform a Taylor expansion around $b_0,$ so that  
\begin{align*}
\fdr{\beta}{{M}_n}(\beta,b_n)&\approx \fdr{\beta}{{M}_n}(\beta,b_0) + (b_n-b_0)^T \sddr{b}{\beta}{{M}_n}(\beta_{n,\gamma},b_0)\\
&= \fdr{\beta}{{M}_n}(\beta,b_0) + o_P(1) \sddr{b}{\beta}{{M}_n}(\beta,b_0)\\
&= \fdr{\beta}{{M}_n}(\beta,b_0) + o_P(1),
\end{align*} 
where the last equality follows due to boundedness of the cross derivative by Assumption \ref{as:mbounded} and the Central Limit Theorem, since the sample cross derivative is an average.
Hence, when we write  $\fdr{\beta}{{M}_n}$ below, we implicitly take this Taylor expansion and only focus on showing convergence of the term  $\fdr{\beta}{{M}_n}(\beta,b_0),$ and we will do the same for the arguments concerning $H_n$ and $X_n$ that follow.

Now we will show that  $\fdr{\beta}{{M}_n}$ converges to the desired expectation. Recall that $z$ is defined in (\ref{eq:z}). Note
\begin{align}
\label{eq:zn}
J_n = \fdr{\beta}{{M}_n} = E_n z
=E_n\left(
\frac{ \fdr{\beta}{v} }{E_n r}
- \frac{v (E_n \fdr{\beta}{r})}{(E_n r)^2} 
- \fdr{\beta}{\gamma  kl(\beta,b)}\right).
\end{align}
Define also \[J_0=\fdr{\beta}{M_0}=E\left(
 \fdr{\beta}{v} 
- \fdr{\beta}{\gamma  kl(\beta,b)}\right).\]
We need to therefore show that $J_n\gop J_0.$  By Remark \ref{rem:Er1}, we have that $E_n r\gop Er=1,$ and by Lemma \ref{lem:limDerivsRn} of Appendix \ref{app:limDerivs}, $E_n \fdr{\beta}{r}\gop 0.$ Hence, we can conclude, based on Slutsky's theorem, that $J_n\gop J_0.$  

Let us now consider the Hessian, $H_n.$
Taking derivatives (the building blocks of which
are shown in Appendix \ref{app:gradients}), we have that 
\begin{align}
\label{eq:Hn}
H_n \notag &= 
\frac{E_n\sdsr{\beta}{v}}{E_n r}
- \frac{E_n\fdr{\beta}v}{E_n r}\frac{(E_n\fdr{\beta}{r})^T}{E_n r}
- \frac{E_n\fdr{\beta}{r}}{E_n r} \frac{(E_n \fdr{\beta}{v}  )^T}{E_n r}\\ 
&-\frac{E_n v}{E_n r} \frac{E_n\sdsr{\beta}{r}}{E_n r} 
+2\frac{E_n v}{E_n r} \frac{E_n \fdr{\beta}{r}}{E_nr} \frac{(E_n \fdr{\beta}{r})^T  }{E_n r}
-\sdsr{\beta}{\gamma   KL_n(\beta,b)}\\ &= (I) + (II) + (III) + (IV) + (V) +(VI).  
\end{align}

Note that by Lemma \ref{lem:limDerivsRn} and Slutsky's theorem, the terms $(II-V)$ converge in probability to zero. The term $E_n r \gop 1$ by Remark \ref{rem:Er1}. Note then that by the Law of Large Numbers, $H_n\gop H_0.$ 

We also have that, similarly,
\begin{align}
\label{eq:Xn}
X_n  \notag &= 
\frac{ E_n \sddr{b}{\beta} v }{E_n r} 
- \frac{E_n \fdr{\beta}{v}}{E_n r}\frac{(E_n \fdr{b}{r} )^T}{E_n r}  
- \frac{E_n \fdr{\beta} r}{E_n r} \frac{(E_n \fdr{b}{v} )^T }{ E_n r }\\  
& - \frac{E_n v}{E_n r} \frac{E_n \sddr{b}{\beta}{ r}  }{E_n r}   
+2 \frac{E_n v}{E_n r}  \frac{E_n \fdr{\beta}{r}}{E_n r}\frac{(E_n\fdr{b}{r})^T  }{E_n r}  
- \sddr{b}{\beta}{\gamma KL_n(\beta,b)}\\
&= (I) + (II) + (III) +(IV) + (V) + (VI).
\end{align}
Note that, as with $H_n,$ by Lemma \ref{lem:limDerivsRn} and Slutsky's theorem, the terms $(II-V)$ converge in probability to zero. The term $E_n r \gop 1$ by Remark \ref{rem:Er1}.  Note then that by the law of Large Numbers and Slutsky's theorem,
$X_n\gop X_0.$

\subsection{Proof of Theorem \ref{thm:betanConsistNorm} (i)}
\label{app:consistproof}
\consistproof

\subsection{Proof of Theorem \ref{thm:betanConsistNorm} (ii)}
\label{app:norprf}
\norprf

\subsection{Proof of Theorem \ref{thmNormalityBetanLambda}}
\label{app:adaprf}
\adaprf

\subsection{Behavioral policy (nuisance) influence function}
\label{nuisinfluencefunction}
\newcommand{\nuisinfluencefunction}{
We derive an estimator for $\sqrt{n}(b_n-b_0)$ that does not depend on $b_0.$ Let $\sqrt{n}(b_n-b_0)\gol q_0,$ which is normally distributed \citep{casella2002statistical,van2000asymptotic}.  Observe, by Taylor's theorem, that
\begin{align*}
\sqrt{n}(b_n-b_0)
&\approx-(E_nl''(b_0))^{-1}\sqrt{n}E_nl'(b_0)\\
&=-\sqrt{n}E_n(E_nl''(b_0))^{-1}l'(b_0)\\
&\approx -\sqrt{n}E_n(E_nl''(b_n))^{-1}l'(b_n)\\
&= \sqrt{n}E_n q = \sqrt{n} q_n,
\end{align*}
where $q$ is defined in (\ref{eq:q})
and $\sqrt{n}q_n$
 does not depend on $b_0,$ as desired.
}
\nuisinfluencefunction


\subsection{Variance of the suggested policy coefficients}
\label{app:avarbeta}
\varDeriv



\subsection{Gradients}
\label{app:gradients}
\grads


\subsection{Estimator for the variance of the value}
\label{app:valuevar}
\valuevar

\subsection{Estimating the true behavior-constrained estimand}
\label{supp.section.estimating.beta.gamma}
We would like to assess coverage now using our estimator for $\sigma^2$, the variance of $\beta_n$. Recall that in the policy search setting, $\beta_{0,\gamma}$ is unknown.  
We thus must estimate $\beta_{0,\gamma}.$ 
In simulation,  since we know the functional form of the reward is (\ref{eq:rewardspec}), we can derive an estimator for $V_0,$ which we call $V_{0,n},$ that avoids importance sampling. 
We do so by finding the maximizer of $V_{0,n},$ as defined below. 
\noindent In our simulations, we defined 
\[R(s_t,a_t,s_{t+1})=-s_{t+1,2} a_t.\]  
Hence, in the multidimensional state, multi-stage case
\begin{align*}
V_0=E_{\beta}\sum_{t=0}^{T}R(S_t,A_t,S_{t+1})&=E\sum_{t=0}^T u^{t-1} R_t\\
&=E\left[\sum_{t=0}^T u^{t-1} E\left[R_t|S_1,\dots,S_t\right]\right] \\
&=-E\left[\sum_{t=0}^T u^{t-1} E\left[(S_{t})_2A_t|S_1,\dots,S_t\right]\right]\\
&=-E\left[\sum_{t=0}^T u^{t-1} E\left[(S_{t})_2A_t|S_t\right]\right] \text{ \ \ \     (since } (A_t\perp S_j| S_t), j<t) \\
&=-E\left[\sum_{t=0}^T u^{t-1} E\left[(S_{t})_2A_t|S_{t,1},\dots,S_{t,K}\right]\right]\\ 
&=-E\left[\sum_{t=0}^T u^{t-1} (S_{t})_2 E\left[A_t|S_t\right]\right] \\
&=-E\left[\sum_{t=0}^T u^{t-1} (S_{t})_2 \expit(\beta^T S_t)\right]\\
&\approx -\frac{1}{n}\sum_{i=1}^n \sum_{t=0}^T u^{t-1} (S_{i,t})_2 \expit(\beta^TS_{i,t}) =: V_{0,n}
\end{align*}
Note that, because we are omniscient when we simulate data, and therefore we know $R_t$, this estimator \textit{does not depend on the ratio defined in (\ref{eq:r})}, which is the destabilizing component of $V_n$. Hence, this estimator will converge arbitrarily quickly. We have chosen to designate this pure Monte-Carlo estimator of $V_0$ as $V_{0,n}$, which is different from our importance sampling estimator of $V_0,$ which we called $V_n.$ In practice, we would never be able to estimate $V_{0,n}$ with real data, but we could estimate $V_n.$  We therefore only use $V_{0,n}$ to check our theoretical results with simulated data. We use a derivation as we did because in simulation we know the true reward function ((\ref{eq:rewardspec})), and therefore we can derive an expression for the expected reward that \textit{avoids the importance sampling} ratio present in $V_n$ from Equation (\ref{eq:vnestw}).

 Define $M_{0,n}=V_{0,n}+\gamma KL_n(\beta,b_n)$  and 
 \begin{equation}
 \label{eq:mc:beta0ngamma}
 \beta_{0,n,\gamma}=\arg\max_{\beta}M_{0,n},
 \end{equation}
 which is an estimator for $\beta_{0,\gamma}$ that we take to be ``true'' because it converges much faster than does $\beta_{n,\gamma}.$  

  We can therefore use the same $n$ for estimating $M_{0,n}$ and ${M}_n$ and get much stabler estimates in the former.  This allows us to then in a sense average over our estimates of $\beta_{0,n,\gamma}$ as $\bar{\beta}_{0,n,\gamma}$, making it even more stable.

\subsection{Simulation coefficients}
\label{sim.coef}
\simcoef


\subsection{Coverage summary statistics}
\label{app:simsum}
\simmetrics

\subsection[Behavioral specification]{Calibration curve of the behavioral policy model in the real data (test set)}
\label{app:calcurve}

A calibration curve for the behavioral policy model is shown in Figure \ref{calcurveInference}.
  We see that the behavioral estimates align with the observed probabilities, which serves as evidence that the logistic regression model for the behavioral policy in Equation (\ref{eq:thetapolicyInf}) is reasonable.

\begin{figure}[htp!]
    \centering
\includegraphics[width=0.5\textwidth]{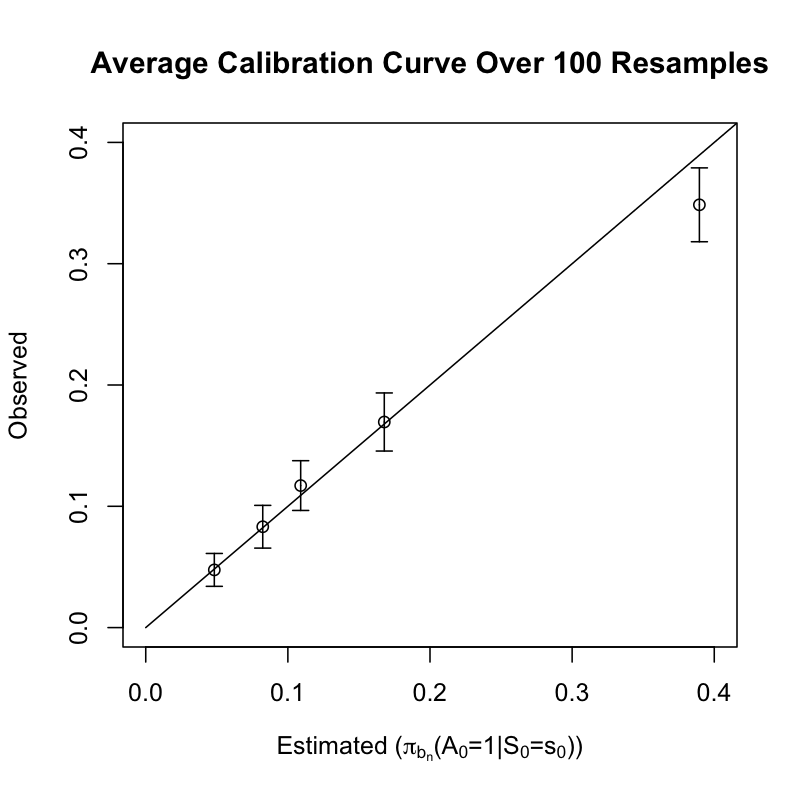}
    \caption[ Real data calibration curve for the behavioral policy]{{\bf  Real data calibration curve for the behavioral policy.} We show a calibration curve for the real data, in which the estimated and observed probabilities are compared. The plot is generated based on the predtools R package. We resample the real data, each time training on one half and then generating a calibration curve for the test data, and then we finally average these test set calibration curves (generated by the R package ``predtools''). The behavioral policy is stationary by Assumption \ref{assum:stationarity}, so we just treat action-state observations at different time steps as independent (the confidence intervals might be too narrow, but mostly we are focused on the fact that the point estimates are approximately near the identity here).
    }
    \label{calcurveInference}
\end{figure}

\end{document}